\newif\ifarxiv
\arxivtrue

\documentclass[conference]{IEEEtran}
\IEEEoverridecommandlockouts

\usepackage[pass]{geometry}
\usepackage{amsthm}
\usepackage{thm-restate}
\usepackage{cite}
\usepackage{subcaption}
\usepackage[cmex10]{amsmath}
\usepackage{amssymb,amsfonts}
\usepackage{bm}
\usepackage{algorithmic}
\usepackage{graphicx}
\usepackage{textcomp}
\usepackage{xcolor}
\usepackage{tabularray}
\usepackage{contour}
\usepackage{mathtools}
\def\BibTeX{{\rm B\kern-.05em{\sc i\kern-.025em b}\kern-.08em
    T\kern-.1667em\lower.7ex\hbox{E}\kern-.125em}}
\usepackage{eqnarray}
\ifarxiv
    \usepackage{hyperref}
\else
\fi

\usepackage[utf8]{inputenc}
\usepackage{mleftright}       %
\mleftright                   %

\usepackage{pgfplots}
 \usetikzlibrary{arrows.meta}
 \usetikzlibrary{backgrounds}
 \usepgfplotslibrary{patchplots}
 \usepgfplotslibrary{fillbetween}
 \pgfplotsset{%
      compat=1.18,
     layers/standard/.define layer set={%
         background,axis background,axis grid,axis ticks,axis lines,axis tick labels,pre main,main,axis descriptions,axis foreground%
     }{
         grid style={/pgfplots/on layer=axis grid},%
         tick style={/pgfplots/on layer=axis ticks},%
         axis line style={/pgfplots/on layer=axis lines},%
         label style={/pgfplots/on layer=axis descriptions},%
         legend style={/pgfplots/on layer=axis descriptions},%
         title style={/pgfplots/on layer=axis descriptions},%
         colorbar style={/pgfplots/on layer=axis descriptions},%
         ticklabel style={/pgfplots/on layer=axis tick labels},%
         axis background@ style={/pgfplots/on layer=axis background},%
         3d box foreground style={/pgfplots/on layer=axis foreground},%
     },
 }

 \pgfplotsset{
 colormap={plots1}{rgb(0.00000000)=(0.30980000,0.10090000,0.23840000)
 rgb(0.00392157)=(0.30690000,0.10310000,0.24320000)
 rgb(0.00784314)=(0.30410000,0.10530000,0.24800000)
 rgb(0.01176471)=(0.30120000,0.10760000,0.25300000)
 rgb(0.01568627)=(0.29840000,0.11020000,0.25800000)
 rgb(0.01960784)=(0.29550000,0.11280000,0.26310000)
 rgb(0.02352941)=(0.29270000,0.11540000,0.26820000)
 rgb(0.02745098)=(0.28980000,0.11830000,0.27350000)
 rgb(0.03137255)=(0.28690000,0.12120000,0.27890000)
 rgb(0.03529412)=(0.28410000,0.12430000,0.28430000)
 rgb(0.03921569)=(0.28120000,0.12750000,0.28990000)
 rgb(0.04313725)=(0.27830000,0.13090000,0.29550000)
 rgb(0.04705882)=(0.27540000,0.13430000,0.30120000)
 rgb(0.05098039)=(0.27250000,0.13790000,0.30700000)
 rgb(0.05490196)=(0.26960000,0.14160000,0.31290000)
 rgb(0.05882353)=(0.26670000,0.14550000,0.31890000)
 rgb(0.06274510)=(0.26370000,0.14940000,0.32490000)
 rgb(0.06666667)=(0.26080000,0.15350000,0.33110000)
 rgb(0.07058824)=(0.25780000,0.15770000,0.33730000)
 rgb(0.07450980)=(0.25490000,0.16210000,0.34360000)
 rgb(0.07843137)=(0.25190000,0.16650000,0.34990000)
 rgb(0.08235294)=(0.24900000,0.17110000,0.35630000)
 rgb(0.08627451)=(0.24600000,0.17590000,0.36280000)
 rgb(0.09019608)=(0.24310000,0.18070000,0.36930000)
 rgb(0.09411765)=(0.24020000,0.18580000,0.37590000)
 rgb(0.09803922)=(0.23730000,0.19080000,0.38250000)
 rgb(0.10196078)=(0.23440000,0.19610000,0.38920000)
 rgb(0.10588235)=(0.23160000,0.20140000,0.39590000)
 rgb(0.10980392)=(0.22870000,0.20690000,0.40260000)
 rgb(0.11372549)=(0.22600000,0.21250000,0.40940000)
 rgb(0.11764706)=(0.22330000,0.21820000,0.41610000)
 rgb(0.12156863)=(0.22070000,0.22410000,0.42290000)
 rgb(0.12549020)=(0.21820000,0.23000000,0.42970000)
 rgb(0.12941176)=(0.21580000,0.23610000,0.43650000)
 rgb(0.13333333)=(0.21340000,0.24220000,0.44330000)
 rgb(0.13725490)=(0.21130000,0.24850000,0.45000000)
 rgb(0.14117647)=(0.20920000,0.25480000,0.45680000)
 rgb(0.14509804)=(0.20740000,0.26130000,0.46350000)
 rgb(0.14901961)=(0.20570000,0.26780000,0.47030000)
 rgb(0.15294118)=(0.20430000,0.27440000,0.47700000)
 rgb(0.15686275)=(0.20300000,0.28110000,0.48360000)
 rgb(0.16078431)=(0.20200000,0.28790000,0.49020000)
 rgb(0.16470588)=(0.20130000,0.29480000,0.49680000)
 rgb(0.16862745)=(0.20080000,0.30170000,0.50340000)
 rgb(0.17254902)=(0.20070000,0.30870000,0.50990000)
 rgb(0.17647059)=(0.20080000,0.31580000,0.51640000)
 rgb(0.18039216)=(0.20130000,0.32290000,0.52280000)
 rgb(0.18431373)=(0.20210000,0.33010000,0.52920000)
 rgb(0.18823529)=(0.20330000,0.33740000,0.53550000)
 rgb(0.19215686)=(0.20480000,0.34470000,0.54180000)
 rgb(0.19607843)=(0.20680000,0.35200000,0.54800000)
 rgb(0.20000000)=(0.20900000,0.35940000,0.55420000)
 rgb(0.20392157)=(0.21170000,0.36690000,0.56030000)
 rgb(0.20784314)=(0.21470000,0.37430000,0.56640000)
 rgb(0.21176471)=(0.21810000,0.38190000,0.57250000)
 rgb(0.21568627)=(0.22180000,0.38940000,0.57850000)
 rgb(0.21960784)=(0.22590000,0.39700000,0.58450000)
 rgb(0.22352941)=(0.23040000,0.40460000,0.59040000)
 rgb(0.22745098)=(0.23530000,0.41230000,0.59620000)
 rgb(0.23137255)=(0.24040000,0.42000000,0.60210000)
 rgb(0.23529412)=(0.24590000,0.42770000,0.60790000)
 rgb(0.23921569)=(0.25170000,0.43540000,0.61360000)
 rgb(0.24313725)=(0.25780000,0.44320000,0.61930000)
 rgb(0.24705882)=(0.26420000,0.45100000,0.62500000)
 rgb(0.25098039)=(0.27090000,0.45870000,0.63060000)
 rgb(0.25490196)=(0.27790000,0.46650000,0.63620000)
 rgb(0.25882353)=(0.28510000,0.47440000,0.64170000)
 rgb(0.26274510)=(0.29260000,0.48220000,0.64720000)
 rgb(0.26666667)=(0.30030000,0.49000000,0.65270000)
 rgb(0.27058824)=(0.30830000,0.49780000,0.65810000)
 rgb(0.27450980)=(0.31650000,0.50560000,0.66350000)
 rgb(0.27843137)=(0.32480000,0.51350000,0.66880000)
 rgb(0.28235294)=(0.33340000,0.52130000,0.67410000)
 rgb(0.28627451)=(0.34220000,0.52910000,0.67930000)
 rgb(0.29019608)=(0.35120000,0.53690000,0.68440000)
 rgb(0.29411765)=(0.36030000,0.54460000,0.68950000)
 rgb(0.29803922)=(0.36960000,0.55240000,0.69460000)
 rgb(0.30196078)=(0.37900000,0.56010000,0.69960000)
 rgb(0.30588235)=(0.38860000,0.56780000,0.70450000)
 rgb(0.30980392)=(0.39830000,0.57540000,0.70930000)
 rgb(0.31372549)=(0.40820000,0.58300000,0.71410000)
 rgb(0.31764706)=(0.41820000,0.59050000,0.71870000)
 rgb(0.32156863)=(0.42830000,0.59800000,0.72330000)
 rgb(0.32549020)=(0.43850000,0.60550000,0.72780000)
 rgb(0.32941176)=(0.44890000,0.61280000,0.73220000)
 rgb(0.33333333)=(0.45930000,0.62010000,0.73640000)
 rgb(0.33725490)=(0.46980000,0.62730000,0.74060000)
 rgb(0.34117647)=(0.48040000,0.63450000,0.74460000)
 rgb(0.34509804)=(0.49100000,0.64150000,0.74850000)
 rgb(0.34901961)=(0.50170000,0.64840000,0.75230000)
 rgb(0.35294118)=(0.51250000,0.65520000,0.75590000)
 rgb(0.35686275)=(0.52330000,0.66200000,0.75930000)
 rgb(0.36078431)=(0.53410000,0.66850000,0.76260000)
 rgb(0.36470588)=(0.54490000,0.67490000,0.76560000)
 rgb(0.36862745)=(0.55580000,0.68120000,0.76850000)
 rgb(0.37254902)=(0.56660000,0.68740000,0.77120000)
 rgb(0.37647059)=(0.57750000,0.69330000,0.77370000)
 rgb(0.38039216)=(0.58830000,0.69910000,0.77600000)
 rgb(0.38431373)=(0.59900000,0.70470000,0.77800000)
 rgb(0.38823529)=(0.60980000,0.71010000,0.77980000)
 rgb(0.39215686)=(0.62040000,0.71530000,0.78140000)
 rgb(0.39607843)=(0.63100000,0.72030000,0.78260000)
 rgb(0.40000000)=(0.64140000,0.72500000,0.78360000)
 rgb(0.40392157)=(0.65180000,0.72950000,0.78440000)
 rgb(0.40784314)=(0.66200000,0.73380000,0.78480000)
 rgb(0.41176471)=(0.67210000,0.73780000,0.78500000)
 rgb(0.41568627)=(0.68200000,0.74150000,0.78480000)
 rgb(0.41960784)=(0.69170000,0.74490000,0.78430000)
 rgb(0.42352941)=(0.70130000,0.74810000,0.78360000)
 rgb(0.42745098)=(0.71060000,0.75100000,0.78250000)
 rgb(0.43137255)=(0.71980000,0.75360000,0.78100000)
 rgb(0.43529412)=(0.72870000,0.75580000,0.77930000)
 rgb(0.43921569)=(0.73730000,0.75780000,0.77720000)
 rgb(0.44313725)=(0.74570000,0.75940000,0.77480000)
 rgb(0.44705882)=(0.75380000,0.76080000,0.77200000)
 rgb(0.45098039)=(0.76160000,0.76180000,0.76900000)
 rgb(0.45490196)=(0.76910000,0.76250000,0.76560000)
 rgb(0.45882353)=(0.77640000,0.76280000,0.76190000)
 rgb(0.46274510)=(0.78320000,0.76290000,0.75780000)
 rgb(0.46666667)=(0.78980000,0.76260000,0.75350000)
 rgb(0.47058824)=(0.79610000,0.76200000,0.74880000)
 rgb(0.47450980)=(0.80200000,0.76110000,0.74390000)
 rgb(0.47843137)=(0.80760000,0.75980000,0.73870000)
 rgb(0.48235294)=(0.81280000,0.75830000,0.73320000)
 rgb(0.48627451)=(0.81770000,0.75650000,0.72740000)
 rgb(0.49019608)=(0.82220000,0.75430000,0.72130000)
 rgb(0.49411765)=(0.82640000,0.75190000,0.71510000)
 rgb(0.49803922)=(0.83030000,0.74920000,0.70860000)
 rgb(0.50196078)=(0.83390000,0.74630000,0.70180000)
 rgb(0.50588235)=(0.83710000,0.74310000,0.69490000)
 rgb(0.50980392)=(0.84000000,0.73960000,0.68770000)
 rgb(0.51372549)=(0.84250000,0.73590000,0.68040000)
 rgb(0.51764706)=(0.84480000,0.73200000,0.67280000)
 rgb(0.52156863)=(0.84670000,0.72780000,0.66520000)
 rgb(0.52549020)=(0.84840000,0.72340000,0.65730000)
 rgb(0.52941176)=(0.84970000,0.71890000,0.64930000)
 rgb(0.53333333)=(0.85080000,0.71410000,0.64120000)
 rgb(0.53725490)=(0.85160000,0.70920000,0.63300000)
 rgb(0.54117647)=(0.85220000,0.70410000,0.62460000)
 rgb(0.54509804)=(0.85250000,0.69880000,0.61620000)
 rgb(0.54901961)=(0.85260000,0.69340000,0.60760000)
 rgb(0.55294118)=(0.85240000,0.68780000,0.59900000)
 rgb(0.55686275)=(0.85200000,0.68210000,0.59030000)
 rgb(0.56078431)=(0.85140000,0.67620000,0.58150000)
 rgb(0.56470588)=(0.85050000,0.67030000,0.57270000)
 rgb(0.56862745)=(0.84950000,0.66420000,0.56380000)
 rgb(0.57254902)=(0.84830000,0.65800000,0.55490000)
 rgb(0.57647059)=(0.84690000,0.65170000,0.54590000)
 rgb(0.58039216)=(0.84530000,0.64530000,0.53690000)
 rgb(0.58431373)=(0.84350000,0.63880000,0.52780000)
 rgb(0.58823529)=(0.84160000,0.63220000,0.51880000)
 rgb(0.59215686)=(0.83950000,0.62550000,0.50970000)
 rgb(0.59607843)=(0.83720000,0.61880000,0.50050000)
 rgb(0.60000000)=(0.83480000,0.61190000,0.49140000)
 rgb(0.60392157)=(0.83230000,0.60500000,0.48230000)
 rgb(0.60784314)=(0.82960000,0.59800000,0.47310000)
 rgb(0.61176471)=(0.82670000,0.59090000,0.46400000)
 rgb(0.61568627)=(0.82370000,0.58380000,0.45480000)
 rgb(0.61960784)=(0.82060000,0.57650000,0.44570000)
 rgb(0.62352941)=(0.81730000,0.56920000,0.43660000)
 rgb(0.62745098)=(0.81390000,0.56190000,0.42750000)
 rgb(0.63137255)=(0.81040000,0.55450000,0.41830000)
 rgb(0.63529412)=(0.80670000,0.54690000,0.40930000)
 rgb(0.63921569)=(0.80290000,0.53940000,0.40020000)
 rgb(0.64313725)=(0.79890000,0.53170000,0.39110000)
 rgb(0.64705882)=(0.79480000,0.52400000,0.38210000)
 rgb(0.65098039)=(0.79060000,0.51620000,0.37310000)
 rgb(0.65490196)=(0.78620000,0.50840000,0.36420000)
 rgb(0.65882353)=(0.78170000,0.50040000,0.35530000)
 rgb(0.66274510)=(0.77710000,0.49240000,0.34640000)
 rgb(0.66666667)=(0.77230000,0.48440000,0.33760000)
 rgb(0.67058824)=(0.76730000,0.47620000,0.32880000)
 rgb(0.67450980)=(0.76220000,0.46800000,0.32010000)
 rgb(0.67843137)=(0.75700000,0.45970000,0.31140000)
 rgb(0.68235294)=(0.75160000,0.45140000,0.30280000)
 rgb(0.68627451)=(0.74600000,0.44300000,0.29430000)
 rgb(0.69019608)=(0.74030000,0.43450000,0.28590000)
 rgb(0.69411765)=(0.73450000,0.42600000,0.27760000)
 rgb(0.69803922)=(0.72850000,0.41740000,0.26940000)
 rgb(0.70196078)=(0.72230000,0.40880000,0.26130000)
 rgb(0.70588235)=(0.71600000,0.40010000,0.25340000)
 rgb(0.70980392)=(0.70960000,0.39140000,0.24550000)
 rgb(0.71372549)=(0.70300000,0.38260000,0.23780000)
 rgb(0.71764706)=(0.69630000,0.37380000,0.23030000)
 rgb(0.72156863)=(0.68940000,0.36500000,0.22290000)
 rgb(0.72549020)=(0.68250000,0.35620000,0.21570000)
 rgb(0.72941176)=(0.67540000,0.34740000,0.20870000)
 rgb(0.73333333)=(0.66820000,0.33850000,0.20190000)
 rgb(0.73725490)=(0.66090000,0.32970000,0.19530000)
 rgb(0.74117647)=(0.65350000,0.32090000,0.18900000)
 rgb(0.74509804)=(0.64600000,0.31220000,0.18290000)
 rgb(0.74901961)=(0.63850000,0.30350000,0.17700000)
 rgb(0.75294118)=(0.63090000,0.29490000,0.17140000)
 rgb(0.75686275)=(0.62330000,0.28630000,0.16600000)
 rgb(0.76078431)=(0.61560000,0.27790000,0.16100000)
 rgb(0.76470588)=(0.60790000,0.26950000,0.15620000)
 rgb(0.76862745)=(0.60020000,0.26120000,0.15160000)
 rgb(0.77254902)=(0.59250000,0.25310000,0.14750000)
 rgb(0.77647059)=(0.58490000,0.24500000,0.14350000)
 rgb(0.78039216)=(0.57720000,0.23720000,0.13990000)
 rgb(0.78431373)=(0.56960000,0.22950000,0.13660000)
 rgb(0.78823529)=(0.56210000,0.22190000,0.13360000)
 rgb(0.79215686)=(0.55460000,0.21450000,0.13080000)
 rgb(0.79607843)=(0.54720000,0.20730000,0.12840000)
 rgb(0.80000000)=(0.53990000,0.20020000,0.12620000)
 rgb(0.80392157)=(0.53270000,0.19340000,0.12430000)
 rgb(0.80784314)=(0.52560000,0.18680000,0.12270000)
 rgb(0.81176471)=(0.51860000,0.18030000,0.12140000)
 rgb(0.81568627)=(0.51170000,0.17410000,0.12030000)
 rgb(0.81960784)=(0.50500000,0.16810000,0.11950000)
 rgb(0.82352941)=(0.49840000,0.16230000,0.11890000)
 rgb(0.82745098)=(0.49190000,0.15670000,0.11850000)
 rgb(0.83137255)=(0.48560000,0.15130000,0.11840000)
 rgb(0.83529412)=(0.47940000,0.14620000,0.11840000)
 rgb(0.83921569)=(0.47330000,0.14130000,0.11870000)
 rgb(0.84313725)=(0.46740000,0.13650000,0.11910000)
 rgb(0.84705882)=(0.46160000,0.13200000,0.11970000)
 rgb(0.85098039)=(0.45600000,0.12770000,0.12050000)
 rgb(0.85490196)=(0.45050000,0.12360000,0.12150000)
 rgb(0.85882353)=(0.44510000,0.11980000,0.12260000)
 rgb(0.86274510)=(0.43990000,0.11620000,0.12390000)
 rgb(0.86666667)=(0.43480000,0.11280000,0.12540000)
 rgb(0.87058824)=(0.42980000,0.10960000,0.12690000)
 rgb(0.87450980)=(0.42500000,0.10650000,0.12860000)
 rgb(0.87843137)=(0.42030000,0.10370000,0.13050000)
 rgb(0.88235294)=(0.41570000,0.10110000,0.13240000)
 rgb(0.88627451)=(0.41120000,0.09870000,0.13450000)
 rgb(0.89019608)=(0.40680000,0.09640000,0.13660000)
 rgb(0.89411765)=(0.40260000,0.09440000,0.13890000)
 rgb(0.89803922)=(0.39840000,0.09250000,0.14130000)
 rgb(0.90196078)=(0.39430000,0.09090000,0.14380000)
 rgb(0.90588235)=(0.39040000,0.08940000,0.14640000)
 rgb(0.90980392)=(0.38650000,0.08810000,0.14900000)
 rgb(0.91372549)=(0.38270000,0.08700000,0.15180000)
 rgb(0.91764706)=(0.37900000,0.08590000,0.15470000)
 rgb(0.92156863)=(0.37530000,0.08510000,0.15760000)
 rgb(0.92549020)=(0.37170000,0.08450000,0.16060000)
 rgb(0.92941176)=(0.36820000,0.08410000,0.16380000)
 rgb(0.93333333)=(0.36480000,0.08370000,0.16700000)
 rgb(0.93725490)=(0.36140000,0.08350000,0.17030000)
 rgb(0.94117647)=(0.35810000,0.08350000,0.17360000)
 rgb(0.94509804)=(0.35480000,0.08360000,0.17710000)
 rgb(0.94901961)=(0.35160000,0.08390000,0.18060000)
 rgb(0.95294118)=(0.34840000,0.08430000,0.18420000)
 rgb(0.95686275)=(0.34530000,0.08480000,0.18790000)
 rgb(0.96078431)=(0.34220000,0.08540000,0.19170000)
 rgb(0.96470588)=(0.33910000,0.08620000,0.19550000)
 rgb(0.96862745)=(0.33610000,0.08720000,0.19940000)
 rgb(0.97254902)=(0.33310000,0.08820000,0.20350000)
 rgb(0.97647059)=(0.33010000,0.08940000,0.20750000)
 rgb(0.98039216)=(0.32720000,0.09070000,0.21170000)
 rgb(0.98431373)=(0.32420000,0.09210000,0.21600000)
 rgb(0.98823529)=(0.32130000,0.09360000,0.22030000)
 rgb(0.99215686)=(0.31840000,0.09530000,0.22470000)
 rgb(0.99607843)=(0.31550000,0.09700000,0.22920000)
 rgb(1.00000000)=(0.31260000,0.09890000,0.23380000)},
 }

 \pgfplotsset{
 colormap={plots1}{rgb(0.00000000)=(0.30980000,0.10090000,0.23840000)
 rgb(0.00392157)=(0.30690000,0.10310000,0.24320000)
 rgb(0.00784314)=(0.30410000,0.10530000,0.24800000)
 rgb(0.01176471)=(0.30120000,0.10760000,0.25300000)
 rgb(0.01568627)=(0.29840000,0.11020000,0.25800000)
 rgb(0.01960784)=(0.29550000,0.11280000,0.26310000)
 rgb(0.02352941)=(0.29270000,0.11540000,0.26820000)
 rgb(0.02745098)=(0.28980000,0.11830000,0.27350000)
 rgb(0.03137255)=(0.28690000,0.12120000,0.27890000)
 rgb(0.03529412)=(0.28410000,0.12430000,0.28430000)
 rgb(0.03921569)=(0.28120000,0.12750000,0.28990000)
 rgb(0.04313725)=(0.27830000,0.13090000,0.29550000)
 rgb(0.04705882)=(0.27540000,0.13430000,0.30120000)
 rgb(0.05098039)=(0.27250000,0.13790000,0.30700000)
 rgb(0.05490196)=(0.26960000,0.14160000,0.31290000)
 rgb(0.05882353)=(0.26670000,0.14550000,0.31890000)
 rgb(0.06274510)=(0.26370000,0.14940000,0.32490000)
 rgb(0.06666667)=(0.26080000,0.15350000,0.33110000)
 rgb(0.07058824)=(0.25780000,0.15770000,0.33730000)
 rgb(0.07450980)=(0.25490000,0.16210000,0.34360000)
 rgb(0.07843137)=(0.25190000,0.16650000,0.34990000)
 rgb(0.08235294)=(0.24900000,0.17110000,0.35630000)
 rgb(0.08627451)=(0.24600000,0.17590000,0.36280000)
 rgb(0.09019608)=(0.24310000,0.18070000,0.36930000)
 rgb(0.09411765)=(0.24020000,0.18580000,0.37590000)
 rgb(0.09803922)=(0.23730000,0.19080000,0.38250000)
 rgb(0.10196078)=(0.23440000,0.19610000,0.38920000)
 rgb(0.10588235)=(0.23160000,0.20140000,0.39590000)
 rgb(0.10980392)=(0.22870000,0.20690000,0.40260000)
 rgb(0.11372549)=(0.22600000,0.21250000,0.40940000)
 rgb(0.11764706)=(0.22330000,0.21820000,0.41610000)
 rgb(0.12156863)=(0.22070000,0.22410000,0.42290000)
 rgb(0.12549020)=(0.21820000,0.23000000,0.42970000)
 rgb(0.12941176)=(0.21580000,0.23610000,0.43650000)
 rgb(0.13333333)=(0.21340000,0.24220000,0.44330000)
 rgb(0.13725490)=(0.21130000,0.24850000,0.45000000)
 rgb(0.14117647)=(0.20920000,0.25480000,0.45680000)
 rgb(0.14509804)=(0.20740000,0.26130000,0.46350000)
 rgb(0.14901961)=(0.20570000,0.26780000,0.47030000)
 rgb(0.15294118)=(0.20430000,0.27440000,0.47700000)
 rgb(0.15686275)=(0.20300000,0.28110000,0.48360000)
 rgb(0.16078431)=(0.20200000,0.28790000,0.49020000)
 rgb(0.16470588)=(0.20130000,0.29480000,0.49680000)
 rgb(0.16862745)=(0.20080000,0.30170000,0.50340000)
 rgb(0.17254902)=(0.20070000,0.30870000,0.50990000)
 rgb(0.17647059)=(0.20080000,0.31580000,0.51640000)
 rgb(0.18039216)=(0.20130000,0.32290000,0.52280000)
 rgb(0.18431373)=(0.20210000,0.33010000,0.52920000)
 rgb(0.18823529)=(0.20330000,0.33740000,0.53550000)
 rgb(0.19215686)=(0.20480000,0.34470000,0.54180000)
 rgb(0.19607843)=(0.20680000,0.35200000,0.54800000)
 rgb(0.20000000)=(0.20900000,0.35940000,0.55420000)
 rgb(0.20392157)=(0.21170000,0.36690000,0.56030000)
 rgb(0.20784314)=(0.21470000,0.37430000,0.56640000)
 rgb(0.21176471)=(0.21810000,0.38190000,0.57250000)
 rgb(0.21568627)=(0.22180000,0.38940000,0.57850000)
 rgb(0.21960784)=(0.22590000,0.39700000,0.58450000)
 rgb(0.22352941)=(0.23040000,0.40460000,0.59040000)
 rgb(0.22745098)=(0.23530000,0.41230000,0.59620000)
 rgb(0.23137255)=(0.24040000,0.42000000,0.60210000)
 rgb(0.23529412)=(0.24590000,0.42770000,0.60790000)
 rgb(0.23921569)=(0.25170000,0.43540000,0.61360000)
 rgb(0.24313725)=(0.25780000,0.44320000,0.61930000)
 rgb(0.24705882)=(0.26420000,0.45100000,0.62500000)
 rgb(0.25098039)=(0.27090000,0.45870000,0.63060000)
 rgb(0.25490196)=(0.27790000,0.46650000,0.63620000)
 rgb(0.25882353)=(0.28510000,0.47440000,0.64170000)
 rgb(0.26274510)=(0.29260000,0.48220000,0.64720000)
 rgb(0.26666667)=(0.30030000,0.49000000,0.65270000)
 rgb(0.27058824)=(0.30830000,0.49780000,0.65810000)
 rgb(0.27450980)=(0.31650000,0.50560000,0.66350000)
 rgb(0.27843137)=(0.32480000,0.51350000,0.66880000)
 rgb(0.28235294)=(0.33340000,0.52130000,0.67410000)
 rgb(0.28627451)=(0.34220000,0.52910000,0.67930000)
 rgb(0.29019608)=(0.35120000,0.53690000,0.68440000)
 rgb(0.29411765)=(0.36030000,0.54460000,0.68950000)
 rgb(0.29803922)=(0.36960000,0.55240000,0.69460000)
 rgb(0.30196078)=(0.37900000,0.56010000,0.69960000)
 rgb(0.30588235)=(0.38860000,0.56780000,0.70450000)
 rgb(0.30980392)=(0.39830000,0.57540000,0.70930000)
 rgb(0.31372549)=(0.40820000,0.58300000,0.71410000)
 rgb(0.31764706)=(0.41820000,0.59050000,0.71870000)
 rgb(0.32156863)=(0.42830000,0.59800000,0.72330000)
 rgb(0.32549020)=(0.43850000,0.60550000,0.72780000)
 rgb(0.32941176)=(0.44890000,0.61280000,0.73220000)
 rgb(0.33333333)=(0.45930000,0.62010000,0.73640000)
 rgb(0.33725490)=(0.46980000,0.62730000,0.74060000)
 rgb(0.34117647)=(0.48040000,0.63450000,0.74460000)
 rgb(0.34509804)=(0.49100000,0.64150000,0.74850000)
 rgb(0.34901961)=(0.50170000,0.64840000,0.75230000)
 rgb(0.35294118)=(0.51250000,0.65520000,0.75590000)
 rgb(0.35686275)=(0.52330000,0.66200000,0.75930000)
 rgb(0.36078431)=(0.53410000,0.66850000,0.76260000)
 rgb(0.36470588)=(0.54490000,0.67490000,0.76560000)
 rgb(0.36862745)=(0.55580000,0.68120000,0.76850000)
 rgb(0.37254902)=(0.56660000,0.68740000,0.77120000)
 rgb(0.37647059)=(0.57750000,0.69330000,0.77370000)
 rgb(0.38039216)=(0.58830000,0.69910000,0.77600000)
 rgb(0.38431373)=(0.59900000,0.70470000,0.77800000)
 rgb(0.38823529)=(0.60980000,0.71010000,0.77980000)
 rgb(0.39215686)=(0.62040000,0.71530000,0.78140000)
 rgb(0.39607843)=(0.63100000,0.72030000,0.78260000)
 rgb(0.40000000)=(0.64140000,0.72500000,0.78360000)
 rgb(0.40392157)=(0.65180000,0.72950000,0.78440000)
 rgb(0.40784314)=(0.66200000,0.73380000,0.78480000)
 rgb(0.41176471)=(0.67210000,0.73780000,0.78500000)
 rgb(0.41568627)=(0.68200000,0.74150000,0.78480000)
 rgb(0.41960784)=(0.69170000,0.74490000,0.78430000)
 rgb(0.42352941)=(0.70130000,0.74810000,0.78360000)
 rgb(0.42745098)=(0.71060000,0.75100000,0.78250000)
 rgb(0.43137255)=(0.71980000,0.75360000,0.78100000)
 rgb(0.43529412)=(0.72870000,0.75580000,0.77930000)
 rgb(0.43921569)=(0.73730000,0.75780000,0.77720000)
 rgb(0.44313725)=(0.74570000,0.75940000,0.77480000)
 rgb(0.44705882)=(0.75380000,0.76080000,0.77200000)
 rgb(0.45098039)=(0.76160000,0.76180000,0.76900000)
 rgb(0.45490196)=(0.76910000,0.76250000,0.76560000)
 rgb(0.45882353)=(0.77640000,0.76280000,0.76190000)
 rgb(0.46274510)=(0.78320000,0.76290000,0.75780000)
 rgb(0.46666667)=(0.78980000,0.76260000,0.75350000)
 rgb(0.47058824)=(0.79610000,0.76200000,0.74880000)
 rgb(0.47450980)=(0.80200000,0.76110000,0.74390000)
 rgb(0.47843137)=(0.80760000,0.75980000,0.73870000)
 rgb(0.48235294)=(0.81280000,0.75830000,0.73320000)
 rgb(0.48627451)=(0.81770000,0.75650000,0.72740000)
 rgb(0.49019608)=(0.82220000,0.75430000,0.72130000)
 rgb(0.49411765)=(0.82640000,0.75190000,0.71510000)
 rgb(0.49803922)=(0.83030000,0.74920000,0.70860000)
 rgb(0.50196078)=(0.83390000,0.74630000,0.70180000)
 rgb(0.50588235)=(0.83710000,0.74310000,0.69490000)
 rgb(0.50980392)=(0.84000000,0.73960000,0.68770000)
 rgb(0.51372549)=(0.84250000,0.73590000,0.68040000)
 rgb(0.51764706)=(0.84480000,0.73200000,0.67280000)
 rgb(0.52156863)=(0.84670000,0.72780000,0.66520000)
 rgb(0.52549020)=(0.84840000,0.72340000,0.65730000)
 rgb(0.52941176)=(0.84970000,0.71890000,0.64930000)
 rgb(0.53333333)=(0.85080000,0.71410000,0.64120000)
 rgb(0.53725490)=(0.85160000,0.70920000,0.63300000)
 rgb(0.54117647)=(0.85220000,0.70410000,0.62460000)
 rgb(0.54509804)=(0.85250000,0.69880000,0.61620000)
 rgb(0.54901961)=(0.85260000,0.69340000,0.60760000)
 rgb(0.55294118)=(0.85240000,0.68780000,0.59900000)
 rgb(0.55686275)=(0.85200000,0.68210000,0.59030000)
 rgb(0.56078431)=(0.85140000,0.67620000,0.58150000)
 rgb(0.56470588)=(0.85050000,0.67030000,0.57270000)
 rgb(0.56862745)=(0.84950000,0.66420000,0.56380000)
 rgb(0.57254902)=(0.84830000,0.65800000,0.55490000)
 rgb(0.57647059)=(0.84690000,0.65170000,0.54590000)
 rgb(0.58039216)=(0.84530000,0.64530000,0.53690000)
 rgb(0.58431373)=(0.84350000,0.63880000,0.52780000)
 rgb(0.58823529)=(0.84160000,0.63220000,0.51880000)
 rgb(0.59215686)=(0.83950000,0.62550000,0.50970000)
 rgb(0.59607843)=(0.83720000,0.61880000,0.50050000)
 rgb(0.60000000)=(0.83480000,0.61190000,0.49140000)
 rgb(0.60392157)=(0.83230000,0.60500000,0.48230000)
 rgb(0.60784314)=(0.82960000,0.59800000,0.47310000)
 rgb(0.61176471)=(0.82670000,0.59090000,0.46400000)
 rgb(0.61568627)=(0.82370000,0.58380000,0.45480000)
 rgb(0.61960784)=(0.82060000,0.57650000,0.44570000)
 rgb(0.62352941)=(0.81730000,0.56920000,0.43660000)
 rgb(0.62745098)=(0.81390000,0.56190000,0.42750000)
 rgb(0.63137255)=(0.81040000,0.55450000,0.41830000)
 rgb(0.63529412)=(0.80670000,0.54690000,0.40930000)
 rgb(0.63921569)=(0.80290000,0.53940000,0.40020000)
 rgb(0.64313725)=(0.79890000,0.53170000,0.39110000)
 rgb(0.64705882)=(0.79480000,0.52400000,0.38210000)
 rgb(0.65098039)=(0.79060000,0.51620000,0.37310000)
 rgb(0.65490196)=(0.78620000,0.50840000,0.36420000)
 rgb(0.65882353)=(0.78170000,0.50040000,0.35530000)
 rgb(0.66274510)=(0.77710000,0.49240000,0.34640000)
 rgb(0.66666667)=(0.77230000,0.48440000,0.33760000)
 rgb(0.67058824)=(0.76730000,0.47620000,0.32880000)
 rgb(0.67450980)=(0.76220000,0.46800000,0.32010000)
 rgb(0.67843137)=(0.75700000,0.45970000,0.31140000)
 rgb(0.68235294)=(0.75160000,0.45140000,0.30280000)
 rgb(0.68627451)=(0.74600000,0.44300000,0.29430000)
 rgb(0.69019608)=(0.74030000,0.43450000,0.28590000)
 rgb(0.69411765)=(0.73450000,0.42600000,0.27760000)
 rgb(0.69803922)=(0.72850000,0.41740000,0.26940000)
 rgb(0.70196078)=(0.72230000,0.40880000,0.26130000)
 rgb(0.70588235)=(0.71600000,0.40010000,0.25340000)
 rgb(0.70980392)=(0.70960000,0.39140000,0.24550000)
 rgb(0.71372549)=(0.70300000,0.38260000,0.23780000)
 rgb(0.71764706)=(0.69630000,0.37380000,0.23030000)
 rgb(0.72156863)=(0.68940000,0.36500000,0.22290000)
 rgb(0.72549020)=(0.68250000,0.35620000,0.21570000)
 rgb(0.72941176)=(0.67540000,0.34740000,0.20870000)
 rgb(0.73333333)=(0.66820000,0.33850000,0.20190000)
 rgb(0.73725490)=(0.66090000,0.32970000,0.19530000)
 rgb(0.74117647)=(0.65350000,0.32090000,0.18900000)
 rgb(0.74509804)=(0.64600000,0.31220000,0.18290000)
 rgb(0.74901961)=(0.63850000,0.30350000,0.17700000)
 rgb(0.75294118)=(0.63090000,0.29490000,0.17140000)
 rgb(0.75686275)=(0.62330000,0.28630000,0.16600000)
 rgb(0.76078431)=(0.61560000,0.27790000,0.16100000)
 rgb(0.76470588)=(0.60790000,0.26950000,0.15620000)
 rgb(0.76862745)=(0.60020000,0.26120000,0.15160000)
 rgb(0.77254902)=(0.59250000,0.25310000,0.14750000)
 rgb(0.77647059)=(0.58490000,0.24500000,0.14350000)
 rgb(0.78039216)=(0.57720000,0.23720000,0.13990000)
 rgb(0.78431373)=(0.56960000,0.22950000,0.13660000)
 rgb(0.78823529)=(0.56210000,0.22190000,0.13360000)
 rgb(0.79215686)=(0.55460000,0.21450000,0.13080000)
 rgb(0.79607843)=(0.54720000,0.20730000,0.12840000)
 rgb(0.80000000)=(0.53990000,0.20020000,0.12620000)
 rgb(0.80392157)=(0.53270000,0.19340000,0.12430000)
 rgb(0.80784314)=(0.52560000,0.18680000,0.12270000)
 rgb(0.81176471)=(0.51860000,0.18030000,0.12140000)
 rgb(0.81568627)=(0.51170000,0.17410000,0.12030000)
 rgb(0.81960784)=(0.50500000,0.16810000,0.11950000)
 rgb(0.82352941)=(0.49840000,0.16230000,0.11890000)
 rgb(0.82745098)=(0.49190000,0.15670000,0.11850000)
 rgb(0.83137255)=(0.48560000,0.15130000,0.11840000)
 rgb(0.83529412)=(0.47940000,0.14620000,0.11840000)
 rgb(0.83921569)=(0.47330000,0.14130000,0.11870000)
 rgb(0.84313725)=(0.46740000,0.13650000,0.11910000)
 rgb(0.84705882)=(0.46160000,0.13200000,0.11970000)
 rgb(0.85098039)=(0.45600000,0.12770000,0.12050000)
 rgb(0.85490196)=(0.45050000,0.12360000,0.12150000)
 rgb(0.85882353)=(0.44510000,0.11980000,0.12260000)
 rgb(0.86274510)=(0.43990000,0.11620000,0.12390000)
 rgb(0.86666667)=(0.43480000,0.11280000,0.12540000)
 rgb(0.87058824)=(0.42980000,0.10960000,0.12690000)
 rgb(0.87450980)=(0.42500000,0.10650000,0.12860000)
 rgb(0.87843137)=(0.42030000,0.10370000,0.13050000)
 rgb(0.88235294)=(0.41570000,0.10110000,0.13240000)
 rgb(0.88627451)=(0.41120000,0.09870000,0.13450000)
 rgb(0.89019608)=(0.40680000,0.09640000,0.13660000)
 rgb(0.89411765)=(0.40260000,0.09440000,0.13890000)
 rgb(0.89803922)=(0.39840000,0.09250000,0.14130000)
 rgb(0.90196078)=(0.39430000,0.09090000,0.14380000)
 rgb(0.90588235)=(0.39040000,0.08940000,0.14640000)
 rgb(0.90980392)=(0.38650000,0.08810000,0.14900000)
 rgb(0.91372549)=(0.38270000,0.08700000,0.15180000)
 rgb(0.91764706)=(0.37900000,0.08590000,0.15470000)
 rgb(0.92156863)=(0.37530000,0.08510000,0.15760000)
 rgb(0.92549020)=(0.37170000,0.08450000,0.16060000)
 rgb(0.92941176)=(0.36820000,0.08410000,0.16380000)
 rgb(0.93333333)=(0.36480000,0.08370000,0.16700000)
 rgb(0.93725490)=(0.36140000,0.08350000,0.17030000)
 rgb(0.94117647)=(0.35810000,0.08350000,0.17360000)
 rgb(0.94509804)=(0.35480000,0.08360000,0.17710000)
 rgb(0.94901961)=(0.35160000,0.08390000,0.18060000)
 rgb(0.95294118)=(0.34840000,0.08430000,0.18420000)
 rgb(0.95686275)=(0.34530000,0.08480000,0.18790000)
 rgb(0.96078431)=(0.34220000,0.08540000,0.19170000)
 rgb(0.96470588)=(0.33910000,0.08620000,0.19550000)
 rgb(0.96862745)=(0.33610000,0.08720000,0.19940000)
 rgb(0.97254902)=(0.33310000,0.08820000,0.20350000)
 rgb(0.97647059)=(0.33010000,0.08940000,0.20750000)
 rgb(0.98039216)=(0.32720000,0.09070000,0.21170000)
 rgb(0.98431373)=(0.32420000,0.09210000,0.21600000)
 rgb(0.98823529)=(0.32130000,0.09360000,0.22030000)
 rgb(0.99215686)=(0.31840000,0.09530000,0.22470000)
 rgb(0.99607843)=(0.31550000,0.09700000,0.22920000)
 rgb(1.00000000)=(0.31260000,0.09890000,0.23380000)},
 }

 \pgfplotsset{
 colormap={plots2}{rgb(0.00000000)=(0.00130000,0.06980000,0.37950000)
 rgb(0.00392157)=(0.00240000,0.07650000,0.38350000)
 rgb(0.00784314)=(0.00330000,0.08310000,0.38750000)
 rgb(0.01176471)=(0.00410000,0.08960000,0.39150000)
 rgb(0.01568627)=(0.00490000,0.09590000,0.39550000)
 rgb(0.01960784)=(0.00560000,0.10230000,0.39940000)
 rgb(0.02352941)=(0.00620000,0.10850000,0.40340000)
 rgb(0.02745098)=(0.00670000,0.11470000,0.40730000)
 rgb(0.03137255)=(0.00710000,0.12080000,0.41130000)
 rgb(0.03529412)=(0.00750000,0.12700000,0.41520000)
 rgb(0.03921569)=(0.00780000,0.13310000,0.41920000)
 rgb(0.04313725)=(0.00810000,0.13910000,0.42310000)
 rgb(0.04705882)=(0.00840000,0.14520000,0.42700000)
 rgb(0.05098039)=(0.00860000,0.15110000,0.43090000)
 rgb(0.05490196)=(0.00880000,0.15710000,0.43480000)
 rgb(0.05882353)=(0.00890000,0.16320000,0.43870000)
 rgb(0.06274510)=(0.00910000,0.16910000,0.44260000)
 rgb(0.06666667)=(0.00920000,0.17510000,0.44650000)
 rgb(0.07058824)=(0.00930000,0.18110000,0.45030000)
 rgb(0.07450980)=(0.00940000,0.18710000,0.45420000)
 rgb(0.07843137)=(0.00940000,0.19300000,0.45810000)
 rgb(0.08235294)=(0.00950000,0.19900000,0.46200000)
 rgb(0.08627451)=(0.00960000,0.20500000,0.46580000)
 rgb(0.09019608)=(0.00960000,0.21100000,0.46970000)
 rgb(0.09411765)=(0.00970000,0.21700000,0.47360000)
 rgb(0.09803922)=(0.00970000,0.22310000,0.47750000)
 rgb(0.10196078)=(0.00980000,0.22910000,0.48140000)
 rgb(0.10588235)=(0.00990000,0.23520000,0.48520000)
 rgb(0.10980392)=(0.01000000,0.24130000,0.48920000)
 rgb(0.11372549)=(0.01010000,0.24740000,0.49310000)
 rgb(0.11764706)=(0.01030000,0.25350000,0.49700000)
 rgb(0.12156863)=(0.01050000,0.25970000,0.50100000)
 rgb(0.12549020)=(0.01080000,0.26590000,0.50490000)
 rgb(0.12941176)=(0.01120000,0.27200000,0.50890000)
 rgb(0.13333333)=(0.01170000,0.27830000,0.51290000)
 rgb(0.13725490)=(0.01230000,0.28460000,0.51700000)
 rgb(0.14117647)=(0.01290000,0.29090000,0.52100000)
 rgb(0.14509804)=(0.01380000,0.29720000,0.52510000)
 rgb(0.14901961)=(0.01480000,0.30360000,0.52920000)
 rgb(0.15294118)=(0.01610000,0.31000000,0.53330000)
 rgb(0.15686275)=(0.01770000,0.31650000,0.53750000)
 rgb(0.16078431)=(0.01960000,0.32300000,0.54170000)
 rgb(0.16470588)=(0.02190000,0.32960000,0.54590000)
 rgb(0.16862745)=(0.02470000,0.33610000,0.55020000)
 rgb(0.17254902)=(0.02800000,0.34280000,0.55450000)
 rgb(0.17647059)=(0.03200000,0.34950000,0.55890000)
 rgb(0.18039216)=(0.03680000,0.35630000,0.56330000)
 rgb(0.18431373)=(0.04220000,0.36320000,0.56780000)
 rgb(0.18823529)=(0.04800000,0.37010000,0.57230000)
 rgb(0.19215686)=(0.05430000,0.37710000,0.57690000)
 rgb(0.19607843)=(0.06100000,0.38410000,0.58160000)
 rgb(0.20000000)=(0.06810000,0.39130000,0.58630000)
 rgb(0.20392157)=(0.07550000,0.39850000,0.59100000)
 rgb(0.20784314)=(0.08320000,0.40570000,0.59590000)
 rgb(0.21176471)=(0.09140000,0.41310000,0.60080000)
 rgb(0.21568627)=(0.09980000,0.42050000,0.60570000)
 rgb(0.21960784)=(0.10860000,0.42800000,0.61070000)
 rgb(0.22352941)=(0.11770000,0.43560000,0.61580000)
 rgb(0.22745098)=(0.12700000,0.44320000,0.62090000)
 rgb(0.23137255)=(0.13670000,0.45090000,0.62610000)
 rgb(0.23529412)=(0.14660000,0.45860000,0.63130000)
 rgb(0.23921569)=(0.15680000,0.46650000,0.63660000)
 rgb(0.24313725)=(0.16720000,0.47430000,0.64190000)
 rgb(0.24705882)=(0.17780000,0.48220000,0.64720000)
 rgb(0.25098039)=(0.18860000,0.49020000,0.65260000)
 rgb(0.25490196)=(0.19960000,0.49820000,0.65800000)
 rgb(0.25882353)=(0.21080000,0.50620000,0.66350000)
 rgb(0.26274510)=(0.22210000,0.51430000,0.66890000)
 rgb(0.26666667)=(0.23360000,0.52230000,0.67440000)
 rgb(0.27058824)=(0.24520000,0.53040000,0.67990000)
 rgb(0.27450980)=(0.25700000,0.53850000,0.68540000)
 rgb(0.27843137)=(0.26890000,0.54660000,0.69090000)
 rgb(0.28235294)=(0.28080000,0.55470000,0.69640000)
 rgb(0.28627451)=(0.29290000,0.56280000,0.70190000)
 rgb(0.29019608)=(0.30500000,0.57090000,0.70740000)
 rgb(0.29411765)=(0.31720000,0.57900000,0.71300000)
 rgb(0.29803922)=(0.32940000,0.58710000,0.71840000)
 rgb(0.30196078)=(0.34170000,0.59510000,0.72390000)
 rgb(0.30588235)=(0.35410000,0.60320000,0.72940000)
 rgb(0.30980392)=(0.36650000,0.61120000,0.73490000)
 rgb(0.31372549)=(0.37890000,0.61920000,0.74030000)
 rgb(0.31764706)=(0.39130000,0.62720000,0.74580000)
 rgb(0.32156863)=(0.40380000,0.63510000,0.75120000)
 rgb(0.32549020)=(0.41620000,0.64300000,0.75660000)
 rgb(0.32941176)=(0.42870000,0.65100000,0.76200000)
 rgb(0.33333333)=(0.44120000,0.65880000,0.76730000)
 rgb(0.33725490)=(0.45370000,0.66670000,0.77270000)
 rgb(0.34117647)=(0.46620000,0.67450000,0.77800000)
 rgb(0.34509804)=(0.47870000,0.68230000,0.78340000)
 rgb(0.34901961)=(0.49120000,0.69010000,0.78870000)
 rgb(0.35294118)=(0.50370000,0.69790000,0.79400000)
 rgb(0.35686275)=(0.51620000,0.70570000,0.79930000)
 rgb(0.36078431)=(0.52870000,0.71340000,0.80450000)
 rgb(0.36470588)=(0.54110000,0.72110000,0.80980000)
 rgb(0.36862745)=(0.55360000,0.72880000,0.81500000)
 rgb(0.37254902)=(0.56610000,0.73640000,0.82020000)
 rgb(0.37647059)=(0.57860000,0.74410000,0.82540000)
 rgb(0.38039216)=(0.59100000,0.75170000,0.83060000)
 rgb(0.38431373)=(0.60350000,0.75930000,0.83580000)
 rgb(0.38823529)=(0.61590000,0.76690000,0.84090000)
 rgb(0.39215686)=(0.62840000,0.77450000,0.84610000)
 rgb(0.39607843)=(0.64080000,0.78200000,0.85110000)
 rgb(0.40000000)=(0.65320000,0.78950000,0.85620000)
 rgb(0.40392157)=(0.66560000,0.79690000,0.86120000)
 rgb(0.40784314)=(0.67810000,0.80440000,0.86620000)
 rgb(0.41176471)=(0.69050000,0.81170000,0.87110000)
 rgb(0.41568627)=(0.70290000,0.81900000,0.87590000)
 rgb(0.41960784)=(0.71530000,0.82630000,0.88060000)
 rgb(0.42352941)=(0.72760000,0.83340000,0.88510000)
 rgb(0.42745098)=(0.74000000,0.84050000,0.88960000)
 rgb(0.43137255)=(0.75240000,0.84740000,0.89380000)
 rgb(0.43529412)=(0.76470000,0.85410000,0.89780000)
 rgb(0.43921569)=(0.77690000,0.86070000,0.90160000)
 rgb(0.44313725)=(0.78910000,0.86700000,0.90500000)
 rgb(0.44705882)=(0.80120000,0.87300000,0.90800000)
 rgb(0.45098039)=(0.81310000,0.87870000,0.91070000)
 rgb(0.45490196)=(0.82490000,0.88410000,0.91280000)
 rgb(0.45882353)=(0.83640000,0.88890000,0.91430000)
 rgb(0.46274510)=(0.84760000,0.89330000,0.91520000)
 rgb(0.46666667)=(0.85850000,0.89710000,0.91540000)
 rgb(0.47058824)=(0.86890000,0.90020000,0.91480000)
 rgb(0.47450980)=(0.87870000,0.90260000,0.91340000)
 rgb(0.47843137)=(0.88800000,0.90430000,0.91120000)
 rgb(0.48235294)=(0.89650000,0.90520000,0.90800000)
 rgb(0.48627451)=(0.90420000,0.90520000,0.90400000)
 rgb(0.49019608)=(0.91120000,0.90440000,0.89910000)
 rgb(0.49411765)=(0.91720000,0.90280000,0.89340000)
 rgb(0.49803922)=(0.92230000,0.90040000,0.88690000)
 rgb(0.50196078)=(0.92650000,0.89720000,0.87970000)
 rgb(0.50588235)=(0.92980000,0.89330000,0.87180000)
 rgb(0.50980392)=(0.93220000,0.88870000,0.86340000)
 rgb(0.51372549)=(0.93390000,0.88360000,0.85450000)
 rgb(0.51764706)=(0.93480000,0.87790000,0.84520000)
 rgb(0.52156863)=(0.93500000,0.87180000,0.83550000)
 rgb(0.52549020)=(0.93460000,0.86530000,0.82560000)
 rgb(0.52941176)=(0.93380000,0.85850000,0.81540000)
 rgb(0.53333333)=(0.93240000,0.85150000,0.80510000)
 rgb(0.53725490)=(0.93070000,0.84420000,0.79470000)
 rgb(0.54117647)=(0.92860000,0.83680000,0.78420000)
 rgb(0.54509804)=(0.92630000,0.82920000,0.77360000)
 rgb(0.54901961)=(0.92380000,0.82150000,0.76300000)
 rgb(0.55294118)=(0.92100000,0.81380000,0.75230000)
 rgb(0.55686275)=(0.91810000,0.80600000,0.74170000)
 rgb(0.56078431)=(0.91520000,0.79820000,0.73100000)
 rgb(0.56470588)=(0.91210000,0.79030000,0.72040000)
 rgb(0.56862745)=(0.90890000,0.78240000,0.70980000)
 rgb(0.57254902)=(0.90570000,0.77450000,0.69920000)
 rgb(0.57647059)=(0.90250000,0.76670000,0.68860000)
 rgb(0.58039216)=(0.89920000,0.75880000,0.67810000)
 rgb(0.58431373)=(0.89600000,0.75100000,0.66760000)
 rgb(0.58823529)=(0.89270000,0.74310000,0.65710000)
 rgb(0.59215686)=(0.88940000,0.73530000,0.64670000)
 rgb(0.59607843)=(0.88610000,0.72760000,0.63630000)
 rgb(0.60000000)=(0.88280000,0.71980000,0.62590000)
 rgb(0.60392157)=(0.87960000,0.71210000,0.61560000)
 rgb(0.60784314)=(0.87630000,0.70440000,0.60540000)
 rgb(0.61176471)=(0.87300000,0.69680000,0.59510000)
 rgb(0.61568627)=(0.86980000,0.68910000,0.58500000)
 rgb(0.61960784)=(0.86660000,0.68150000,0.57480000)
 rgb(0.62352941)=(0.86330000,0.67400000,0.56470000)
 rgb(0.62745098)=(0.86010000,0.66650000,0.55470000)
 rgb(0.63137255)=(0.85690000,0.65900000,0.54470000)
 rgb(0.63529412)=(0.85370000,0.65150000,0.53480000)
 rgb(0.63921569)=(0.85060000,0.64410000,0.52480000)
 rgb(0.64313725)=(0.84740000,0.63670000,0.51500000)
 rgb(0.64705882)=(0.84430000,0.62930000,0.50510000)
 rgb(0.65098039)=(0.84110000,0.62200000,0.49540000)
 rgb(0.65490196)=(0.83800000,0.61470000,0.48560000)
 rgb(0.65882353)=(0.83490000,0.60740000,0.47590000)
 rgb(0.66274510)=(0.83180000,0.60010000,0.46630000)
 rgb(0.66666667)=(0.82870000,0.59290000,0.45670000)
 rgb(0.67058824)=(0.82560000,0.58580000,0.44710000)
 rgb(0.67450980)=(0.82260000,0.57860000,0.43760000)
 rgb(0.67843137)=(0.81950000,0.57150000,0.42810000)
 rgb(0.68235294)=(0.81650000,0.56440000,0.41870000)
 rgb(0.68627451)=(0.81350000,0.55730000,0.40930000)
 rgb(0.69019608)=(0.81040000,0.55030000,0.39990000)
 rgb(0.69411765)=(0.80740000,0.54330000,0.39060000)
 rgb(0.69803922)=(0.80440000,0.53630000,0.38130000)
 rgb(0.70196078)=(0.80150000,0.52930000,0.37200000)
 rgb(0.70588235)=(0.79850000,0.52240000,0.36280000)
 rgb(0.70980392)=(0.79550000,0.51550000,0.35370000)
 rgb(0.71372549)=(0.79250000,0.50860000,0.34450000)
 rgb(0.71764706)=(0.78960000,0.50170000,0.33540000)
 rgb(0.72156863)=(0.78660000,0.49480000,0.32630000)
 rgb(0.72549020)=(0.78370000,0.48800000,0.31730000)
 rgb(0.72941176)=(0.78070000,0.48110000,0.30830000)
 rgb(0.73333333)=(0.77770000,0.47430000,0.29930000)
 rgb(0.73725490)=(0.77480000,0.46750000,0.29040000)
 rgb(0.74117647)=(0.77180000,0.46060000,0.28140000)
 rgb(0.74509804)=(0.76880000,0.45380000,0.27250000)
 rgb(0.74901961)=(0.76580000,0.44690000,0.26360000)
 rgb(0.75294118)=(0.76270000,0.44010000,0.25480000)
 rgb(0.75686275)=(0.75960000,0.43310000,0.24590000)
 rgb(0.76078431)=(0.75650000,0.42620000,0.23700000)
 rgb(0.76470588)=(0.75330000,0.41920000,0.22820000)
 rgb(0.76862745)=(0.75010000,0.41220000,0.21930000)
 rgb(0.77254902)=(0.74670000,0.40500000,0.21050000)
 rgb(0.77647059)=(0.74320000,0.39780000,0.20160000)
 rgb(0.78039216)=(0.73970000,0.39050000,0.19270000)
 rgb(0.78431373)=(0.73590000,0.38310000,0.18390000)
 rgb(0.78823529)=(0.73200000,0.37550000,0.17500000)
 rgb(0.79215686)=(0.72790000,0.36770000,0.16600000)
 rgb(0.79607843)=(0.72350000,0.35990000,0.15710000)
 rgb(0.80000000)=(0.71890000,0.35180000,0.14820000)
 rgb(0.80392157)=(0.71400000,0.34350000,0.13930000)
 rgb(0.80784314)=(0.70880000,0.33500000,0.13050000)
 rgb(0.81176471)=(0.70330000,0.32640000,0.12150000)
 rgb(0.81568627)=(0.69740000,0.31750000,0.11280000)
 rgb(0.81960784)=(0.69120000,0.30850000,0.10410000)
 rgb(0.82352941)=(0.68470000,0.29930000,0.09560000)
 rgb(0.82745098)=(0.67770000,0.28990000,0.08740000)
 rgb(0.83137255)=(0.67050000,0.28050000,0.07920000)
 rgb(0.83529412)=(0.66290000,0.27100000,0.07150000)
 rgb(0.83921569)=(0.65500000,0.26150000,0.06410000)
 rgb(0.84313725)=(0.64700000,0.25210000,0.05710000)
 rgb(0.84705882)=(0.63870000,0.24270000,0.05060000)
 rgb(0.85098039)=(0.63030000,0.23350000,0.04480000)
 rgb(0.85490196)=(0.62170000,0.22440000,0.03940000)
 rgb(0.85882353)=(0.61310000,0.21570000,0.03480000)
 rgb(0.86274510)=(0.60450000,0.20710000,0.03110000)
 rgb(0.86666667)=(0.59590000,0.19870000,0.02820000)
 rgb(0.87058824)=(0.58740000,0.19070000,0.02600000)
 rgb(0.87450980)=(0.57890000,0.18290000,0.02440000)
 rgb(0.87843137)=(0.57050000,0.17540000,0.02330000)
 rgb(0.88235294)=(0.56230000,0.16820000,0.02250000)
 rgb(0.88627451)=(0.55410000,0.16120000,0.02210000)
 rgb(0.89019608)=(0.54600000,0.15440000,0.02190000)
 rgb(0.89411765)=(0.53800000,0.14790000,0.02170000)
 rgb(0.89803922)=(0.53020000,0.14150000,0.02170000)
 rgb(0.90196078)=(0.52240000,0.13530000,0.02180000)
 rgb(0.90588235)=(0.51480000,0.12920000,0.02200000)
 rgb(0.90980392)=(0.50720000,0.12330000,0.02220000)
 rgb(0.91372549)=(0.49970000,0.11750000,0.02250000)
 rgb(0.91764706)=(0.49230000,0.11180000,0.02280000)
 rgb(0.92156863)=(0.48500000,0.10620000,0.02310000)
 rgb(0.92549020)=(0.47780000,0.10060000,0.02350000)
 rgb(0.92941176)=(0.47060000,0.09520000,0.02390000)
 rgb(0.93333333)=(0.46350000,0.08970000,0.02430000)
 rgb(0.93725490)=(0.45650000,0.08430000,0.02480000)
 rgb(0.94117647)=(0.44950000,0.07870000,0.02520000)
 rgb(0.94509804)=(0.44260000,0.07340000,0.02560000)
 rgb(0.94901961)=(0.43570000,0.06790000,0.02610000)
 rgb(0.95294118)=(0.42890000,0.06240000,0.02650000)
 rgb(0.95686275)=(0.42210000,0.05680000,0.02700000)
 rgb(0.96078431)=(0.41540000,0.05110000,0.02740000)
 rgb(0.96470588)=(0.40880000,0.04540000,0.02780000)
 rgb(0.96862745)=(0.40210000,0.03940000,0.02820000)
 rgb(0.97254902)=(0.39560000,0.03340000,0.02860000)
 rgb(0.97647059)=(0.38900000,0.02780000,0.02890000)
 rgb(0.98039216)=(0.38250000,0.02260000,0.02930000)
 rgb(0.98431373)=(0.37600000,0.01760000,0.02960000)
 rgb(0.98823529)=(0.36960000,0.01290000,0.02990000)
 rgb(0.99215686)=(0.36320000,0.00820000,0.03010000)
 rgb(0.99607843)=(0.35680000,0.00400000,0.03030000)
 rgb(1.00000000)=(0.35040000,0.00010000,0.03050000)},
 }
 \pgfplotsset{
 colormap={plots2}{rgb(0.00000000)=(0.00130000,0.06980000,0.37950000)
 rgb(0.00392157)=(0.00240000,0.07650000,0.38350000)
 rgb(0.00784314)=(0.00330000,0.08310000,0.38750000)
 rgb(0.01176471)=(0.00410000,0.08960000,0.39150000)
 rgb(0.01568627)=(0.00490000,0.09590000,0.39550000)
 rgb(0.01960784)=(0.00560000,0.10230000,0.39940000)
 rgb(0.02352941)=(0.00620000,0.10850000,0.40340000)
 rgb(0.02745098)=(0.00670000,0.11470000,0.40730000)
 rgb(0.03137255)=(0.00710000,0.12080000,0.41130000)
 rgb(0.03529412)=(0.00750000,0.12700000,0.41520000)
 rgb(0.03921569)=(0.00780000,0.13310000,0.41920000)
 rgb(0.04313725)=(0.00810000,0.13910000,0.42310000)
 rgb(0.04705882)=(0.00840000,0.14520000,0.42700000)
 rgb(0.05098039)=(0.00860000,0.15110000,0.43090000)
 rgb(0.05490196)=(0.00880000,0.15710000,0.43480000)
 rgb(0.05882353)=(0.00890000,0.16320000,0.43870000)
 rgb(0.06274510)=(0.00910000,0.16910000,0.44260000)
 rgb(0.06666667)=(0.00920000,0.17510000,0.44650000)
 rgb(0.07058824)=(0.00930000,0.18110000,0.45030000)
 rgb(0.07450980)=(0.00940000,0.18710000,0.45420000)
 rgb(0.07843137)=(0.00940000,0.19300000,0.45810000)
 rgb(0.08235294)=(0.00950000,0.19900000,0.46200000)
 rgb(0.08627451)=(0.00960000,0.20500000,0.46580000)
 rgb(0.09019608)=(0.00960000,0.21100000,0.46970000)
 rgb(0.09411765)=(0.00970000,0.21700000,0.47360000)
 rgb(0.09803922)=(0.00970000,0.22310000,0.47750000)
 rgb(0.10196078)=(0.00980000,0.22910000,0.48140000)
 rgb(0.10588235)=(0.00990000,0.23520000,0.48520000)
 rgb(0.10980392)=(0.01000000,0.24130000,0.48920000)
 rgb(0.11372549)=(0.01010000,0.24740000,0.49310000)
 rgb(0.11764706)=(0.01030000,0.25350000,0.49700000)
 rgb(0.12156863)=(0.01050000,0.25970000,0.50100000)
 rgb(0.12549020)=(0.01080000,0.26590000,0.50490000)
 rgb(0.12941176)=(0.01120000,0.27200000,0.50890000)
 rgb(0.13333333)=(0.01170000,0.27830000,0.51290000)
 rgb(0.13725490)=(0.01230000,0.28460000,0.51700000)
 rgb(0.14117647)=(0.01290000,0.29090000,0.52100000)
 rgb(0.14509804)=(0.01380000,0.29720000,0.52510000)
 rgb(0.14901961)=(0.01480000,0.30360000,0.52920000)
 rgb(0.15294118)=(0.01610000,0.31000000,0.53330000)
 rgb(0.15686275)=(0.01770000,0.31650000,0.53750000)
 rgb(0.16078431)=(0.01960000,0.32300000,0.54170000)
 rgb(0.16470588)=(0.02190000,0.32960000,0.54590000)
 rgb(0.16862745)=(0.02470000,0.33610000,0.55020000)
 rgb(0.17254902)=(0.02800000,0.34280000,0.55450000)
 rgb(0.17647059)=(0.03200000,0.34950000,0.55890000)
 rgb(0.18039216)=(0.03680000,0.35630000,0.56330000)
 rgb(0.18431373)=(0.04220000,0.36320000,0.56780000)
 rgb(0.18823529)=(0.04800000,0.37010000,0.57230000)
 rgb(0.19215686)=(0.05430000,0.37710000,0.57690000)
 rgb(0.19607843)=(0.06100000,0.38410000,0.58160000)
 rgb(0.20000000)=(0.06810000,0.39130000,0.58630000)
 rgb(0.20392157)=(0.07550000,0.39850000,0.59100000)
 rgb(0.20784314)=(0.08320000,0.40570000,0.59590000)
 rgb(0.21176471)=(0.09140000,0.41310000,0.60080000)
 rgb(0.21568627)=(0.09980000,0.42050000,0.60570000)
 rgb(0.21960784)=(0.10860000,0.42800000,0.61070000)
 rgb(0.22352941)=(0.11770000,0.43560000,0.61580000)
 rgb(0.22745098)=(0.12700000,0.44320000,0.62090000)
 rgb(0.23137255)=(0.13670000,0.45090000,0.62610000)
 rgb(0.23529412)=(0.14660000,0.45860000,0.63130000)
 rgb(0.23921569)=(0.15680000,0.46650000,0.63660000)
 rgb(0.24313725)=(0.16720000,0.47430000,0.64190000)
 rgb(0.24705882)=(0.17780000,0.48220000,0.64720000)
 rgb(0.25098039)=(0.18860000,0.49020000,0.65260000)
 rgb(0.25490196)=(0.19960000,0.49820000,0.65800000)
 rgb(0.25882353)=(0.21080000,0.50620000,0.66350000)
 rgb(0.26274510)=(0.22210000,0.51430000,0.66890000)
 rgb(0.26666667)=(0.23360000,0.52230000,0.67440000)
 rgb(0.27058824)=(0.24520000,0.53040000,0.67990000)
 rgb(0.27450980)=(0.25700000,0.53850000,0.68540000)
 rgb(0.27843137)=(0.26890000,0.54660000,0.69090000)
 rgb(0.28235294)=(0.28080000,0.55470000,0.69640000)
 rgb(0.28627451)=(0.29290000,0.56280000,0.70190000)
 rgb(0.29019608)=(0.30500000,0.57090000,0.70740000)
 rgb(0.29411765)=(0.31720000,0.57900000,0.71300000)
 rgb(0.29803922)=(0.32940000,0.58710000,0.71840000)
 rgb(0.30196078)=(0.34170000,0.59510000,0.72390000)
 rgb(0.30588235)=(0.35410000,0.60320000,0.72940000)
 rgb(0.30980392)=(0.36650000,0.61120000,0.73490000)
 rgb(0.31372549)=(0.37890000,0.61920000,0.74030000)
 rgb(0.31764706)=(0.39130000,0.62720000,0.74580000)
 rgb(0.32156863)=(0.40380000,0.63510000,0.75120000)
 rgb(0.32549020)=(0.41620000,0.64300000,0.75660000)
 rgb(0.32941176)=(0.42870000,0.65100000,0.76200000)
 rgb(0.33333333)=(0.44120000,0.65880000,0.76730000)
 rgb(0.33725490)=(0.45370000,0.66670000,0.77270000)
 rgb(0.34117647)=(0.46620000,0.67450000,0.77800000)
 rgb(0.34509804)=(0.47870000,0.68230000,0.78340000)
 rgb(0.34901961)=(0.49120000,0.69010000,0.78870000)
 rgb(0.35294118)=(0.50370000,0.69790000,0.79400000)
 rgb(0.35686275)=(0.51620000,0.70570000,0.79930000)
 rgb(0.36078431)=(0.52870000,0.71340000,0.80450000)
 rgb(0.36470588)=(0.54110000,0.72110000,0.80980000)
 rgb(0.36862745)=(0.55360000,0.72880000,0.81500000)
 rgb(0.37254902)=(0.56610000,0.73640000,0.82020000)
 rgb(0.37647059)=(0.57860000,0.74410000,0.82540000)
 rgb(0.38039216)=(0.59100000,0.75170000,0.83060000)
 rgb(0.38431373)=(0.60350000,0.75930000,0.83580000)
 rgb(0.38823529)=(0.61590000,0.76690000,0.84090000)
 rgb(0.39215686)=(0.62840000,0.77450000,0.84610000)
 rgb(0.39607843)=(0.64080000,0.78200000,0.85110000)
 rgb(0.40000000)=(0.65320000,0.78950000,0.85620000)
 rgb(0.40392157)=(0.66560000,0.79690000,0.86120000)
 rgb(0.40784314)=(0.67810000,0.80440000,0.86620000)
 rgb(0.41176471)=(0.69050000,0.81170000,0.87110000)
 rgb(0.41568627)=(0.70290000,0.81900000,0.87590000)
 rgb(0.41960784)=(0.71530000,0.82630000,0.88060000)
 rgb(0.42352941)=(0.72760000,0.83340000,0.88510000)
 rgb(0.42745098)=(0.74000000,0.84050000,0.88960000)
 rgb(0.43137255)=(0.75240000,0.84740000,0.89380000)
 rgb(0.43529412)=(0.76470000,0.85410000,0.89780000)
 rgb(0.43921569)=(0.77690000,0.86070000,0.90160000)
 rgb(0.44313725)=(0.78910000,0.86700000,0.90500000)
 rgb(0.44705882)=(0.80120000,0.87300000,0.90800000)
 rgb(0.45098039)=(0.81310000,0.87870000,0.91070000)
 rgb(0.45490196)=(0.82490000,0.88410000,0.91280000)
 rgb(0.45882353)=(0.83640000,0.88890000,0.91430000)
 rgb(0.46274510)=(0.84760000,0.89330000,0.91520000)
 rgb(0.46666667)=(0.85850000,0.89710000,0.91540000)
 rgb(0.47058824)=(0.86890000,0.90020000,0.91480000)
 rgb(0.47450980)=(0.87870000,0.90260000,0.91340000)
 rgb(0.47843137)=(0.88800000,0.90430000,0.91120000)
 rgb(0.48235294)=(0.89650000,0.90520000,0.90800000)
 rgb(0.48627451)=(0.90420000,0.90520000,0.90400000)
 rgb(0.49019608)=(0.91120000,0.90440000,0.89910000)
 rgb(0.49411765)=(0.91720000,0.90280000,0.89340000)
 rgb(0.49803922)=(0.92230000,0.90040000,0.88690000)
 rgb(0.50196078)=(0.92650000,0.89720000,0.87970000)
 rgb(0.50588235)=(0.92980000,0.89330000,0.87180000)
 rgb(0.50980392)=(0.93220000,0.88870000,0.86340000)
 rgb(0.51372549)=(0.93390000,0.88360000,0.85450000)
 rgb(0.51764706)=(0.93480000,0.87790000,0.84520000)
 rgb(0.52156863)=(0.93500000,0.87180000,0.83550000)
 rgb(0.52549020)=(0.93460000,0.86530000,0.82560000)
 rgb(0.52941176)=(0.93380000,0.85850000,0.81540000)
 rgb(0.53333333)=(0.93240000,0.85150000,0.80510000)
 rgb(0.53725490)=(0.93070000,0.84420000,0.79470000)
 rgb(0.54117647)=(0.92860000,0.83680000,0.78420000)
 rgb(0.54509804)=(0.92630000,0.82920000,0.77360000)
 rgb(0.54901961)=(0.92380000,0.82150000,0.76300000)
 rgb(0.55294118)=(0.92100000,0.81380000,0.75230000)
 rgb(0.55686275)=(0.91810000,0.80600000,0.74170000)
 rgb(0.56078431)=(0.91520000,0.79820000,0.73100000)
 rgb(0.56470588)=(0.91210000,0.79030000,0.72040000)
 rgb(0.56862745)=(0.90890000,0.78240000,0.70980000)
 rgb(0.57254902)=(0.90570000,0.77450000,0.69920000)
 rgb(0.57647059)=(0.90250000,0.76670000,0.68860000)
 rgb(0.58039216)=(0.89920000,0.75880000,0.67810000)
 rgb(0.58431373)=(0.89600000,0.75100000,0.66760000)
 rgb(0.58823529)=(0.89270000,0.74310000,0.65710000)
 rgb(0.59215686)=(0.88940000,0.73530000,0.64670000)
 rgb(0.59607843)=(0.88610000,0.72760000,0.63630000)
 rgb(0.60000000)=(0.88280000,0.71980000,0.62590000)
 rgb(0.60392157)=(0.87960000,0.71210000,0.61560000)
 rgb(0.60784314)=(0.87630000,0.70440000,0.60540000)
 rgb(0.61176471)=(0.87300000,0.69680000,0.59510000)
 rgb(0.61568627)=(0.86980000,0.68910000,0.58500000)
 rgb(0.61960784)=(0.86660000,0.68150000,0.57480000)
 rgb(0.62352941)=(0.86330000,0.67400000,0.56470000)
 rgb(0.62745098)=(0.86010000,0.66650000,0.55470000)
 rgb(0.63137255)=(0.85690000,0.65900000,0.54470000)
 rgb(0.63529412)=(0.85370000,0.65150000,0.53480000)
 rgb(0.63921569)=(0.85060000,0.64410000,0.52480000)
 rgb(0.64313725)=(0.84740000,0.63670000,0.51500000)
 rgb(0.64705882)=(0.84430000,0.62930000,0.50510000)
 rgb(0.65098039)=(0.84110000,0.62200000,0.49540000)
 rgb(0.65490196)=(0.83800000,0.61470000,0.48560000)
 rgb(0.65882353)=(0.83490000,0.60740000,0.47590000)
 rgb(0.66274510)=(0.83180000,0.60010000,0.46630000)
 rgb(0.66666667)=(0.82870000,0.59290000,0.45670000)
 rgb(0.67058824)=(0.82560000,0.58580000,0.44710000)
 rgb(0.67450980)=(0.82260000,0.57860000,0.43760000)
 rgb(0.67843137)=(0.81950000,0.57150000,0.42810000)
 rgb(0.68235294)=(0.81650000,0.56440000,0.41870000)
 rgb(0.68627451)=(0.81350000,0.55730000,0.40930000)
 rgb(0.69019608)=(0.81040000,0.55030000,0.39990000)
 rgb(0.69411765)=(0.80740000,0.54330000,0.39060000)
 rgb(0.69803922)=(0.80440000,0.53630000,0.38130000)
 rgb(0.70196078)=(0.80150000,0.52930000,0.37200000)
 rgb(0.70588235)=(0.79850000,0.52240000,0.36280000)
 rgb(0.70980392)=(0.79550000,0.51550000,0.35370000)
 rgb(0.71372549)=(0.79250000,0.50860000,0.34450000)
 rgb(0.71764706)=(0.78960000,0.50170000,0.33540000)
 rgb(0.72156863)=(0.78660000,0.49480000,0.32630000)
 rgb(0.72549020)=(0.78370000,0.48800000,0.31730000)
 rgb(0.72941176)=(0.78070000,0.48110000,0.30830000)
 rgb(0.73333333)=(0.77770000,0.47430000,0.29930000)
 rgb(0.73725490)=(0.77480000,0.46750000,0.29040000)
 rgb(0.74117647)=(0.77180000,0.46060000,0.28140000)
 rgb(0.74509804)=(0.76880000,0.45380000,0.27250000)
 rgb(0.74901961)=(0.76580000,0.44690000,0.26360000)
 rgb(0.75294118)=(0.76270000,0.44010000,0.25480000)
 rgb(0.75686275)=(0.75960000,0.43310000,0.24590000)
 rgb(0.76078431)=(0.75650000,0.42620000,0.23700000)
 rgb(0.76470588)=(0.75330000,0.41920000,0.22820000)
 rgb(0.76862745)=(0.75010000,0.41220000,0.21930000)
 rgb(0.77254902)=(0.74670000,0.40500000,0.21050000)
 rgb(0.77647059)=(0.74320000,0.39780000,0.20160000)
 rgb(0.78039216)=(0.73970000,0.39050000,0.19270000)
 rgb(0.78431373)=(0.73590000,0.38310000,0.18390000)
 rgb(0.78823529)=(0.73200000,0.37550000,0.17500000)
 rgb(0.79215686)=(0.72790000,0.36770000,0.16600000)
 rgb(0.79607843)=(0.72350000,0.35990000,0.15710000)
 rgb(0.80000000)=(0.71890000,0.35180000,0.14820000)
 rgb(0.80392157)=(0.71400000,0.34350000,0.13930000)
 rgb(0.80784314)=(0.70880000,0.33500000,0.13050000)
 rgb(0.81176471)=(0.70330000,0.32640000,0.12150000)
 rgb(0.81568627)=(0.69740000,0.31750000,0.11280000)
 rgb(0.81960784)=(0.69120000,0.30850000,0.10410000)
 rgb(0.82352941)=(0.68470000,0.29930000,0.09560000)
 rgb(0.82745098)=(0.67770000,0.28990000,0.08740000)
 rgb(0.83137255)=(0.67050000,0.28050000,0.07920000)
 rgb(0.83529412)=(0.66290000,0.27100000,0.07150000)
 rgb(0.83921569)=(0.65500000,0.26150000,0.06410000)
 rgb(0.84313725)=(0.64700000,0.25210000,0.05710000)
 rgb(0.84705882)=(0.63870000,0.24270000,0.05060000)
 rgb(0.85098039)=(0.63030000,0.23350000,0.04480000)
 rgb(0.85490196)=(0.62170000,0.22440000,0.03940000)
 rgb(0.85882353)=(0.61310000,0.21570000,0.03480000)
 rgb(0.86274510)=(0.60450000,0.20710000,0.03110000)
 rgb(0.86666667)=(0.59590000,0.19870000,0.02820000)
 rgb(0.87058824)=(0.58740000,0.19070000,0.02600000)
 rgb(0.87450980)=(0.57890000,0.18290000,0.02440000)
 rgb(0.87843137)=(0.57050000,0.17540000,0.02330000)
 rgb(0.88235294)=(0.56230000,0.16820000,0.02250000)
 rgb(0.88627451)=(0.55410000,0.16120000,0.02210000)
 rgb(0.89019608)=(0.54600000,0.15440000,0.02190000)
 rgb(0.89411765)=(0.53800000,0.14790000,0.02170000)
 rgb(0.89803922)=(0.53020000,0.14150000,0.02170000)
 rgb(0.90196078)=(0.52240000,0.13530000,0.02180000)
 rgb(0.90588235)=(0.51480000,0.12920000,0.02200000)
 rgb(0.90980392)=(0.50720000,0.12330000,0.02220000)
 rgb(0.91372549)=(0.49970000,0.11750000,0.02250000)
 rgb(0.91764706)=(0.49230000,0.11180000,0.02280000)
 rgb(0.92156863)=(0.48500000,0.10620000,0.02310000)
 rgb(0.92549020)=(0.47780000,0.10060000,0.02350000)
 rgb(0.92941176)=(0.47060000,0.09520000,0.02390000)
 rgb(0.93333333)=(0.46350000,0.08970000,0.02430000)
 rgb(0.93725490)=(0.45650000,0.08430000,0.02480000)
 rgb(0.94117647)=(0.44950000,0.07870000,0.02520000)
 rgb(0.94509804)=(0.44260000,0.07340000,0.02560000)
 rgb(0.94901961)=(0.43570000,0.06790000,0.02610000)
 rgb(0.95294118)=(0.42890000,0.06240000,0.02650000)
 rgb(0.95686275)=(0.42210000,0.05680000,0.02700000)
 rgb(0.96078431)=(0.41540000,0.05110000,0.02740000)
 rgb(0.96470588)=(0.40880000,0.04540000,0.02780000)
 rgb(0.96862745)=(0.40210000,0.03940000,0.02820000)
 rgb(0.97254902)=(0.39560000,0.03340000,0.02860000)
 rgb(0.97647059)=(0.38900000,0.02780000,0.02890000)
 rgb(0.98039216)=(0.38250000,0.02260000,0.02930000)
 rgb(0.98431373)=(0.37600000,0.01760000,0.02960000)
 rgb(0.98823529)=(0.36960000,0.01290000,0.02990000)
 rgb(0.99215686)=(0.36320000,0.00820000,0.03010000)
 rgb(0.99607843)=(0.35680000,0.00400000,0.03030000)
 rgb(1.00000000)=(0.35040000,0.00010000,0.03050000)},
 }

 \pgfplotsset{
 colormap={tableaucolorblind}{rgb(0.00000000)=(0.06666667,0.43921569,0.66666667)
 rgb(0.11111111)=(0.98823529,0.49019608,0.04313725)
 rgb(0.22222222)=(0.63921569,0.67450980,0.72549020)
 rgb(0.33333333)=(0.34117647,0.37647059,0.42352941)
 rgb(0.44444444)=(0.37254902,0.63529412,0.80784314)
 rgb(0.55555556)=(0.78431373,0.32156863,0.00000000)
 rgb(0.66666667)=(0.48235294,0.51764706,0.56078431)
 rgb(0.77777778)=(0.63921569,0.80000000,0.91372549)
 rgb(0.88888889)=(1.00000000,0.73725490,0.47450980)
 rgb(1.00000000)=(0.78431373,0.81568627,0.85098039)},
 }

\pgfplotsset{ layers/my layer set/.define layer set={
        background, backishground, main, foreground
    }{
    },
    set layers=my layer set,
}

\usepackage{tikz}

\usepackage{xspace}
\usepackage{tabularx}
\usepackage{booktabs}
\usepackage{arydshln}
\usepackage[normalem]{ulem}
\usepackage{placeins}
\usepackage[capitalize]{cleveref}
\usepackage{siunitx}
\usepackage{intcalc}
\usepackage[inline]{enumitem}


\definecolor{DarkGreen}{rgb}{0.1,0.5,0.1}
\definecolor{DarkRed}{rgb}{0.5,0.1,0.1}
\definecolor{DarkBlue}{rgb}{0.1,0.1,0.5}
\definecolor{bleudefrance}{rgb}{0.19, 0.55, 0.91}

\newtheorem{theorem}{Theorem}

\newtheorem{lemma}{Lemma}

\newtheorem{definition}{Definition}

\newtheorem{remark}{Remark}
\newtheorem{construction}{Construction}

\newcommand{\Fp}{\mathbb{F}_p}

\newcommand{\as}{\bm{\alpha}^{(s)}}
\newcommand{\ap}{\bm{\alpha}^{(p)}}
\newcommand{\bs}{\bm{\beta}^{(s)}}
\newcommand{\bp}{\bm{\beta}^{(p)}}
\newcommand{\ase}{\bm{\alpha}^{(s,E)}}
\newcommand{\ape}{\bm{\alpha}^{(p,E)}}
\newcommand{\bse}{\bm{\beta}^{(s,E)}}
\newcommand{\bpe}{\bm{\beta}^{(p,E)}}

\newcommand{\eas}{\alpha^{(s)}}
\newcommand{\eap}{\alpha^{(p)}}
\newcommand{\ebs}{\beta^{(s)}}
\newcommand{\ebp}{\beta^{(p)}}

\newcommand{\ca}{\zeta}

\newcommand{\kb}{\bar{k}}
\newcommand{\lb}{\bar{\ell}}
\newcommand{\kbp}{\bar{k^\prime}}
\newcommand{\lbp}{\bar{\ell^\prime}}

\newcommand{\Zq}{\ensuremath{\mathbb{Z}_q}\xspace}

\newcommand{\asdtd}{\bm{\alpha}^{(s,D\to D)}}
\newcommand{\apdtd}{\bm{\alpha}^{(p,D \to D)}}
\newcommand{\bsdtd}{\bm{\beta}^{(s,D \to D)}}
\newcommand{\bpdtd}{\bm{\beta}^{(p,D \to D)}}

\newcommand{\asctc}{\bm{\alpha}^{(s,C\to C)}}
\newcommand{\apctc}{\bm{\alpha}^{(p,C\to C)}}
\newcommand{\bsctc}{\bm{\beta}^{(s, C\to C)}}
\newcommand{\bpctc}{\bm{\beta}^{(p, C\to C)}}

\newcommand{\asdtc}{\bm{\alpha}^{(s,D\to C)}}
\newcommand{\apdtc}{\bm{\alpha}^{(p,D\to C)}}
\newcommand{\bsdtc}{\bm{\beta}^{(s, D\to C)}}
\newcommand{\bpdtc}{\bm{\beta}^{(p, D\to C)}}

\newcommand{\gap}{\mathrm{gap}\xspace}

\newcommand{\gasp}{GASP\xspace}
\newcommand{\gaspr}{GASP\textsubscript{r}\xspace}
\newcommand{\gaspsmall}{GASP\textsubscript{small}\xspace}
\newcommand{\gaspbig}{GASP\textsubscript{big}\xspace}

\newcommand{\dogrs}{DOG\textsubscript{rs}\xspace}
\newcommand{\catx}{CAT\textsubscript{x}\xspace}
\newcommand{\ggasp}{GGASP\xspace}
\newcommand{\ggaspr}{GGASP\textsubscript{r}\xspace}

\newcommand{\dog}{DOG\xspace}

\newcommand{\bfA}{\ensuremath{\bm{A}}\xspace}

\newcommand{\bfB}{\ensuremath{\bm{B}}\xspace}

\newcommand{\bfR}{\ensuremath{\bm{R}}\xspace}
\newcommand{\bfS}{\ensuremath{\bm{S}}\xspace}

\newcommand{\bfF}{\ensuremath{\bm{F}}\xspace}
\newcommand{\bfG}{\ensuremath{\bm{G}}\xspace}
\newcommand{\bfH}{\ensuremath{\bm{H}}\xspace}
\newcommand{\bfV}{\ensuremath{\bm{V}}\xspace}

\newcommand{\cA}{\ensuremath{\mathcal{A}}\xspace}
\newcommand{\cB}{\ensuremath{\mathcal{B}}\xspace}
\newcommand{\cC}{\ensuremath{\mathcal{C}}\xspace}
\newcommand{\cD}{\ensuremath{\mathcal{D}}\xspace}

\newcommand{\cX}{\ensuremath{\mathcal{X}}\xspace}

\newcommand{\cU}{\ensuremath{\mathcal{U}}\xspace}

\newcommand{\cAo}{\ensuremath{\mathcal{A}^\mathrm{OPP}}\xspace}

\newcommand{\cCo}{\ensuremath{\mathcal{C}^\mathrm{OPP}}\xspace}
\newcommand{\cUo}{\ensuremath{\mathcal{U}^\mathrm{OPP}}\xspace}
\newcommand{\TRo}{\ensuremath{\mathcal{TR}^\mathrm{OPP}}\xspace}
\newcommand{\BRo}{\ensuremath{\mathcal{BR}^\mathrm{OPP}}\xspace}

\newcommand{\TLt}{\ensuremath{\widetilde{\mathcal{TL}}}\xspace}
\newcommand{\TL}{\ensuremath{\mathcal{TL}}\xspace}
\newcommand{\TR}{\ensuremath{\mathcal{TR}}\xspace}
\newcommand{\BL}{\ensuremath{\mathcal{BL}}\xspace}
\newcommand{\BR}{\ensuremath{\mathcal{BR}}\xspace}

\newcommand{\zbl}{\ensuremath{z_{\mathrm{BL}}}\xspace}
\newcommand{\ztr}{\ensuremath{z_{\mathrm{TR}}}\xspace}
\newcommand{\zbr}{\ensuremath{z_{\mathrm{BR}}}\xspace}

\newcommand{\tmpvar}{\ensuremath{\chi}\xspace}
\newcommand{\tmpvarbl}{\ensuremath{\xi}\xspace}
\newcommand{\tmpvarbr}{\ensuremath{\eta}\xspace}

\newcommand{\setint}[2]{\{#1 \!:\! #2\}}
\newcommand{\vecint}[2]{(#1 \!:\! #2)}

\newcommand{\define}{\ensuremath{\triangleq}}
\newcommand{\defeq}{\ensuremath{\triangleq}}

\newcommand{\st}{\,\vert\,}
\newcommand{\vect}{\operatorname{vec}}

\begin{document}

\title{On the Extension of Private Distributed Matrix Multiplication Schemes to the Grid Partition\thanks{This work is supported by the German Research Foundation (DFG) under Grant No.\ BI~2492/1-1.}}

\author{
   \IEEEauthorblockN{Christoph Hofmeister\textsuperscript{1}, Razane Tajeddine\textsuperscript{2}, Antonia Wachter-Zeh\textsuperscript{1}, and Rawad Bitar\textsuperscript{1}}\\
  \IEEEauthorblockA{\textsuperscript{1}Technical University of Munich (TUM), \{christoph.hofmeister, antonia.wachter-zeh, rawad.bitar\}@tum.de}
  \IEEEauthorblockA{\textsuperscript{2}American University of Beirut (AUB), razane.tajeddine@aub.edu.lb}
\vspace{-3ex}
}

\maketitle

\begin{abstract} 
  We consider polynomial codes for private distributed matrix multiplication (PDMM/SDMM). 
  Existing codes for PDMM are either specialized for the outer product partitioning (OPP), or inner product partitioning (IPP), or are valid for the more general grid partitioning (GP).
  We design extension operations that can be applied to a large class of OPP code designs to extend them to the GP case. Applying them to existing codes improves upon the state-of-the-art for certain parameters.
 Additionally, we show that the GP schemes resulting from extension fulfill additional combinatorial constraints, potentially limiting their performance. We illustrate this point by presenting a new GP scheme that does not adhere to these constraints and outperforms the state-of-the-art for a range of parameters.
\end{abstract}

\section{Introduction}

We consider the problem of a main node distributing the multiplication of two private matrices $\bfA$ and $\bfB$ to $N$ honest-but-curious worker nodes. The workers carry out the required computation, but up to $T$ of them may share information to learn more about the private matrices, where $T$ is called the collusion parameter. The goal is to design computational tasks to be sent to as few worker nodes as possible, ensuring: \begin{enumerate*}[label={\emph{(\roman*)}}]
    \item \emph{information-theoretic privacy} of the matrices $\bfA$ and $\bfB$; and \item \emph{decodability} of the product $\bfA \bfB$ upon receiving the answers of the workers' computations.
\end{enumerate*} This setting, referred to as \emph{private/secure distributed matrix multiplication} (PDMM/SDMM), has been studied in recent years in the literature, e.g., \cite{chang2018capacity,chang2019upload,doliveira2021degreec,makkonen2023algebraica,doliveira2020gaspa,aliasgari2020private,makkonen2024general,karpuk2024modular,machado2022root}.

\begin{figure}[t]
  \resizebox{0.9\linewidth}{!}{
  \input{figures/fishplot_40_20.tex}}
  \caption{ \small This graph shows which known scheme requires the fewest workers for a given $2\leq K=L \leq 40$, $2\leq M \leq 40$, and $T=20$. Circles mark GP schemes from the literature, namely ROU~\cite{machado2022root}, BGK~\cite{byrne2023straggler}, MP~\cite{karpuk2024modular}, and \ggaspr~\cite{karpuk2024modular}. Triangles mark a new extension of the existing OPP scheme \dogrs~\cite{hofmeister2025cat} to the GP. Squares mark the scheme presented in \cref{sec:newscheme}.
      The x-axis represents how many blocks of $\bfA$ and $\bfB$ need to be multiplied. The y-axis indicates how wide or tall $\bfA$ and $\bfB$ are as block matrices. Points on the x-axis correspond to cases where $K=M=L$. Intuitively, the further up a point is, the closer it is to OPP rather than IPP. \vspace{-2em}}
      \label{fig:fishplot}
\end{figure}

Typically, the computational tasks are designed as follows. The matrix $\bfA$ is split into $K \times M$ equally sized blocks and $\bfB$ into $M \times L$ equally sized blocks. Those blocks are then embedded as coefficients of two polynomials $\bm{F}(x)$ and $\bm{G}(x)$ such that all the blocks of the product $\bfA \bfB$ appear as coefficients of the product polynomial $\bm{H}(x) \define \bm{F}(x)\bm{G}(x)$. As such, each worker receives an evaluation of $\bm{F}(x)$ and $\bm{G}(x)$ that they multiply and send back to the main node. Upon receiving enough computation results from the workers, the product $\bm{A}\bm{B}$ can be decoded. 
The main task in this setting is the design of the polynomials $\bfF(x)$ and $\bm{G}(x)$ and the parameters $K,M$ and $L$, which determine a tradeoff between main-node-to-worker communication (upload cost), worker-to-main-node communication (download cost), per-worker computation cost, and the total number of workers $N$ required. They are further influenced by how wide or tall $\bfA$ and $\bfB$ are. 
In this work, we consider $K, M, L$, and the collusion parameter $T$ to be given and aim to design polynomials $\bfF(x)$ and $\bfG(x)$ that minimize $N$, the total number of workers required for the multiplication.

\emph{Related work:} This problem has received significant attention in the literature. We mention a few closely related works and refer the reader to references within the cited papers. In \cite{doliveira2020gaspa,doliveira2021degreec,hofmeister2025cat}, the polynomials are designed for the special case $M=1$, called the \emph{outer product partitioning} (OPP) case. In \cite{lopez2022secure,mital2022secure,machado2023hera, makkonen2024flexible}, codes for the special case $K=L=1$, called \emph{inner product partitioning} (IPP), were presented. The general case, referred to as \emph{general partitioning} or \emph{grid partitioning} (GP), has been considered in~\cite{aliasgari2020private,makkonen2024general,karpuk2024modular,byrne2023straggler}. An interesting observation made in \cite{doliveira2020gaspa} is that the value of $N$ is not only affected by the total degree of $\bfH(x)$ (prior to \cite{doliveira2020gaspa}, $N$ was usually taken as $\mathrm{deg}(\bfH)+1$), but can also be reduced by using the number of zero coefficients in it, when certain conditions hold. This observation led to constructions with low values of $N$ in the OPP case~\cite{doliveira2020gaspa,doliveira2021degreec} and was used in the GP case~\cite{byrne2023straggler}. Leveraging the properties of roots of unity can lower the number of required workers~\cite{mital2022secure,karpuk2024modular,hofmeister2025cat}. %

\pagebreak
\emph{Contributions:} We focus on the GP case and design novel codes that require fewer workers than existing schemes in many parameter regimes. To achieve this, our approach builds on the concept of degree tables (DTs) introduced in \cite{doliveira2021degreec} and cyclic-addition degree tables (CATs) introduced in \cite{hofmeister2025cat}, which are specifically designed for OPP. %
In particular, we present the following contributions: \begin{itemize}
  \item Inspired by the relation of GGASP \cite{karpuk2024modular} to \gaspr \cite{doliveira2021degreec} we formally define three extension operations that can be used to extend a large class of DTs and CATs to GP.
      \item Applying the extension operations to the OPP schemes \catx and \dogrs from~\cite{hofmeister2025cat}, yields new GP schemes which outperform the state-of-the-art for a wide range of parameters. 
      \item We show that the grid extension operations impose additional unnecessary constraints on the resulting DTs/CATs. Leveraging this fact, we design GP-CATs tailored specifically to GP, which significantly outperform existing schemes, and those designed above, across a wide range of parameters.
    \end{itemize}
    A comparison of the newly proposed schemes and the literature for a slice of the parameter space is shown in \cref{fig:fishplot}.

We limit the scope of this work to the number of workers used by a scheme for given $K, M, L$, and $T$.
Other considerations, not addressed in this work, are, for instance, communication-efficient recovery of the result from a higher number of worker nodes, or straggler tolerance, \emph{i.e.}, sending tasks to a larger number $N^\prime > N$ of workers, of which any subset of size $N$ is sufficient to obtain the result.
Additionally, all schemes operate in finite fields and differ in the field sizes they support. We take the position that some large prime fields should be supported to allow for approximate multiplication of real matrices. This is true for the schemes we design.

\section{Preliminaries} \label{sec:prelims}

\emph{Notation:}
Bold lowercase letters denote column vectors and bold uppercase letters denote matrices. For vectors $\bm{x}$ and $\bm{y}$ of lengths $n_x$ and $n_y$, let $\bm{x} \oplus \bm{y}$ denote $(x_1+y_1, x_1+y_2, \dots, x_1+y_{n_y}, x_2+y_1, x_2+y_2, \dots, x_2+y_{n_y}, \dots, x_{n_x}+y_{n_y})^T$.
Let $\bfV(\bm{x}, \bm{y})$ denote the generalized Vandermonde matrix $\bfV(\bm{x}, \bm{y}) \defeq (x_i^{y_j})_{1\leq i \leq n_x, 1\leq j \leq n_y}$, \emph{i.e.}, the $n_x \times n_y$ matrix with entry $(i,j)$ equal to $x_i^{y_j}$. 
For integers $a$ and $b$, define the vector $\vecint{a}{b} \defeq (a, a+1, \dots, b)^T$ and the set $\setint{a}{b} \defeq \{a, a+1, \dots, b\}$. 
For vectors $\bm{x}$ and $\bm{y}$, $\bm{x} || \bm{y}$ denotes their concatenation. Let $\bm{x}_{a:b}$ denote $(x_a, x_{a+1}, \dots, x_{a+b})^T$. For sets $\cA$ and $\cB$, define the sumset $\cA+\cB\defeq\{a+b | a\in \cA, b\in \cB\}$. Further, let $a+\cB\defeq \{a+b|b\in \cB\}$ and $a\cdot\cB\defeq\{a\cdot b | b \in \cB\}$. For a vector $\bm{x}= (x_1, \dots, x_{n_x})^T$ of length $n_x$, let $\{\bm{x}\} \defeq \{x_i | i \in \setint{1}{n_x}\}$ and conversely for a set $\cX$ of integers, let $\vect(\cX)$ denote the vector containing the elements of $\cX$ in increasing order.
Let $a\pm \bm{x} \defeq (a\pm x_1,\dots, a\pm x_{n_x})$.
We define the following type of generalized arithmetic progression with common differences $1$ and $x$ as
  \begin{align*}
    \gap(\ell, x, r) \defeq (&0, 1, 2, \dots, r-1, \\
    &x, x+1, x+2, \dots, x+r-1,\\ 
                        &2x, 2x+1, 2x + 2, \dots, 2x+r-1, \dots),
  \end{align*}
  such that $\gap(\ell, x, r)\in \mathbb{Z}^\ell$, where we call $r$ its chain length.

\subsection{(Cyclic-Addition) Degree Tables for PDMM}
The matrices $\bfA$ and $\bfB$ are supposed to be generated from a finite alphabet, typically a finite field $\Fp$.
Given the parameters $K, M, L$ and $T$, the matrix $\bfA$ and the matrix $\bfB$ are split into equal-sized blocks $\bfA_{i,j}$ and $\bfB_{i,j}$, respectively. Then, $T$ matrices $\bfR_1,\cdots,\bfR_T$, each of the same dimension as $\bfA_{i,j}$ and $T$ matrices $\bfS_1,\cdots,\bfS_T$ each of the same dimension as $\bfB_{i,j}$ are generated. The entries of $\bfR_\ell$ and $\bfS_\ell$ for $\ell\in\{1,\cdots T\}$ are drawn independently and uniformly at random from $\Fp$.

The majority of schemes for PDMM use a polynomial encoding of the form 
\begin{align}
  \bfF(x) &= \sum_{i=1}^{K} \sum_{j=1}^M \bfA_{i,j} x^{\alpha^{(p)}_{(i-1)M + j}} + \sum_{i=1}^{T} \bfR_i x^{\alpha^{(s)}_{i}}, \label{eq:encF} \\
\bfG(x) & = \sum_{i=1}^{L} \sum_{j=1}^M \bfB_{i,j} x^{\beta^{(p)}_{(i-1)M+j}} + \sum_{i=1}^{T} \bfS_i x^{\beta^{(s)}_{i}}, \label{eq:encG}
\end{align} where the degrees are entries of carefully chosen degree vectors
    $\ap \in \mathbb{Z}^{K\cdot M}, \bp \in \mathbb{Z}^{L \cdot M}, \as \in \mathbb{Z}^{T},$ and $\bs\in \mathbb{Z}^T$. In addition, the evaluation points $\boldsymbol{\rho} \in \mathbb{F}_p^{N}$ are carefully chosen. 
    We restrict our focus to these encodings.

For the OPP schemes, \emph{i.e.}, for $M=1$, the work in \cite{doliveira2020gaspa,doliveira2021degreec} focused on reducing the number of non-zero coefficients in $\bfH(x)$. To that end, the authors proposed the use of an addition table, termed \emph{degree table}, to understand the implications of the choice of vectors $\ap, \bp, \as$ and $\bs$ that satisfy the privacy and decodability constraints on the number of non-zero coefficients in $\bfH(x)$. Let $\boldsymbol{\alpha}$ and $\boldsymbol{\beta}$ be the concatenation of $\ap||\as$ and $\bp||\bs$, respectively. The degree table is a pictorial representation of the vector $\boldsymbol{\alpha}\oplus\boldsymbol{\beta}$. The number of distinct entries in $\boldsymbol{\alpha}\oplus\boldsymbol{\beta}$ maps to the number of non-zero coefficients in $\bfH(x)$, and is therefore equal to $N$~\cite{doliveira2020gaspa}. 

The goal becomes to choose $\boldsymbol{\alpha}$ and $\boldsymbol{\beta}$ in a way to maintain the privacy and decodability constraints and reduce the number of distinct entries in $\boldsymbol{\alpha}\oplus\boldsymbol{\beta}$. For ease of notation, the degree table is divided into four quadrants: \begin{enumerate*}[label={\emph{(\roman*)}}]
    \item the top-left quadrant with entries pertaining to $\ap\oplus\bp$;
    \item the top-right quadrant with entries of $\ap\oplus\bs$;
    \item the bottom-left quadrant with entries of $\as\oplus\bp$; and
    \item the bottom-right quadrant with entries of $\as\oplus\bs$.
\end{enumerate*} We denote by $\TL, \TR, \BL$ and $\BR$ the sets of entries pertaining to each quadrant. A valid DT and CAT puts conditions on the entries in $\TL, \TR, \BL$ and $\BR$ as explained in detail in~\cite{hofmeister2025cat}.

  The degree table has been adapted to algebraic geometry codes (PoleNumber tables) in \cite{makkonen2023algebraica} and to evaluations at roots of unity (cyclic-addition degree tables) in \cite{hofmeister2025cat}. Choosing evaluation points at roots of unity allows for a modulo wrap-around in the degrees of $x$ in $\bfH(x)$, which translates to a modulo wrap-around in the entries of the degree table. We also note that in a CAT, all additions are in $\Zq$. A depiction of a CAT is given in \cref{fig:cat422}, where the evaluation points are chosen to be powers of a primitive $29$th root of unity. With $K=6$, $L=3$, $M=1$, $T = 2$, we choose $\ap = (0:5)$, $\as = (6,28)$, $\bp=(0,22,15)$ and $\bs=(7,8)$. The entries of the vector $\boldsymbol{\alpha}\oplus\boldsymbol{\beta} \mod 29$ are shown in the table.

  The notions of a DT and CAT extend straightforwardly to the GP, which we call GPDT and GPCAT, respectively. The constraints on the entries of $\ap$, $\as$, $\bp$, and $\bs$ remain the same, but the constraints on the entries of the DT or the CAT change. To explain GPDT and GPCAT, we need some additional notation. In GP, the vectors $\ap$ and $\bp$ are split into $K$, resp.\ $L$, vectors $\ap_1,\cdots,\ap_K$, resp.\ $\bp_1,\cdots,\bp_L$, each of length $M$. The division of $\ap\oplus \bp$ into sums of $\ap_k\oplus \bp_\ell$, $k\in \{1:K\}, \ell \in \{1:L\}$, divides the top left quadrant of the degree table into $K\times L$ blocks, each of size $M\times M$ (pictorially placed as $K$ rows and $L$ columns); see \cref{fig:gcat422}. For ease of notation, we denote by $\TL_{k, \ell}$ the set of entries pertaining to the $(k,\ell)$th block, i.e. $\ap_k\oplus \bp_\ell$. Let $\eap_k[i]$ and $\ebp_\ell[j]$ be the $i$th entry of $\ap_k$ and the $j$th entry of $\bp_\ell$, respectively. 
  We denote by $\cU_{k,\ell} \defeq \{\eap_k[i] + \ebp_\ell[M-i+1] \st i\in\{1:M\}\}$ the set of entries pertaining to the main anti-diagonal of the $(k,\ell)$th block. %
  Further, let $\cA_{k,\ell} \defeq \{\eap_k[i] + \ebp_\ell[j] \st i,j\in \setint{1}{M}, i+j < M+1\}$ and $\cB_{k,\ell} \defeq \{\eap_k[i] + \ebp_\ell[j] \st i,j\in \setint{1}{M}, i+j > M+1\}$ be the sets of entries of the $(k, \ell)$th block above and below the antidiagonal, respectively.
  We denote by $\TLt_{k,\ell} \defeq \cA_{k,\ell} \cup \cB_{k,\ell}$ the set of off-antidiagonal entries and note that $\TL_{k,\ell} = \cA_{k,\ell} \cup \cU_{k,\ell} \cup \cB_{k,\ell}$.
  The preceding definitions are illustrated in \cref{fig:illustrextGP}.

  For a given $M>1$ and an OPP-DT/CAT for parameters $(K^\mathrm{OPP}=KM, L^\mathrm{OPP}=L, T^\mathrm{OPP}=T)$ specifically, we divide the top-left quadrant in a similar way as for GP-DTs/CATs and split $\ap$ into $K$ equally sized vectors $\ap_1, \dots, \ap_{K}$. Define $\cAo_{i, \ell} \defeq \{\eap_k[i] + \ebp_\ell \st i\in\setint{1}{M-1}\}$ and the singleton $\cUo_{k,\ell} \defeq \{\eap_k[M] + \ebp_\ell\}$. Similarly, $\cCo_\ell \defeq \{\as + \ebp_\ell \}$.
  The preceding definitions are illustrated in \cref{fig:illustrextOPP}.

 For instance, in \cref{fig:gcat422}, $\ap_1 = (0,1,2)$, $\bp_2= (22,23,24)$, $\TL_{1,2} = \setint{22}{26}$ and $\cU_{1,2} = \{24\}$. In GP, instead of requiring each entry of $\TL$ to be unique within the table, the only required constraint is that the entries of each $\cU_{k,\ell}$ for all $k\in \setint{1}{K}$ and $\ell \in \setint{1}{L}$ cannot appear elsewhere in the table.   We provide next a formal definition of GPDT and GPCAT. %

\begin{figure}
    \begin{subfigure}{0.31\linewidth}
      \centering
      \input{figures/CAT_6_3_2.tex}
      \vspace{-0.4cm}
      \caption{\small (6, 3, 2)-CAT.}
      \label{fig:cat422}
    \end{subfigure}
    \begin{subfigure}{0.68\linewidth}
      \centering
      \input{figures/GCAT_2_3_3_2.tex}
      \caption{\small (2, 3, 3, 2)-GPCAT.}
      \label{fig:gcat422}
    \end{subfigure}
    \caption{\small The \catx scheme~\cite{hofmeister2025cat} for $K=6, M=1, L=3$ and $T=2$ using powers of the $29$th root of unity as evaluation points shown in (a) is extended to the GP with $K=2$, $M=3$, $L=3$ and $T=2$ shown in (b). The first column of each table corresponds to the vector $\boldsymbol{\alpha}$ and the first row to $\boldsymbol{\beta}$.}
\end{figure}

  \begin{definition}[GP Degree Table (GPDT) and GP Cyclic Addition DT (GPCAT)] \label{def:gptables}
    Consider vectors $\ap \in \Fp^{KM}$, $\bp \in \Fp^{ML}$, $\as \in \Fp^{T}$, and $\bs\in \Fp^{T}$, let $\TL = \{\ap\} + \{\bp\}$, $\TR = \{\ap\} + \{\bs\}$, $\BL = \{\as\} + \{\bp\}$, $\BR = \{\as\} + \{\bs\}$ and define $\TL_{k,\ell}$, $\cU_{k,\ell}$ and $\TLt_{k,\ell}$ as above. 
    
    \emph{\underline{GP Degree Table:}} A $(K,M,L, T)$-GP degree table that satisfies $T$-privacy and decodability from $N$ workers requires the following properties:
  \begin{enumerate}[label=\Roman*)]
    \item \label{property:gp1} $|\TL \cup \TR \cup \BL \cup \BR| = N$,  %
    \item \label{property:gp2} $\forall k, k^\prime \in \setint{0}{K-1}, \ell, \ell^\prime \in \setint{0}{L-1}, k\neq k^\prime \vee \ell \neq \ell^\prime: a) \cU_{k,\ell} \cap \cU_{k^\prime, \ell^\prime} = \emptyset $, b) $\cU_{k,\ell} \cap \TR = \emptyset$, c) $\cU_{k,\ell} \cap \BL = \emptyset$, d) $\cU_{k,\ell} \cap \BR = \emptyset$, and e) $\cU_{k,\ell} \cap \TLt_{k^\prime, \ell^\prime} = \cU_{k,\ell} \cap \TLt_{k, \ell} = \emptyset$,
    \item $|\{\ap || \as\}| = KM+T$, and $|\{\bp||\bs\}| = LM+T$.
  \end{enumerate}

  \emph{\underline{GP Cyclic Addition DT:}}
  A $(K, M, L, T)$-GP cyclic-addition degree table (GPCAT) that satisfies $T$-privacy and decodability from $N$ workers requires properties \ref{property:gp1} and \ref{property:gp2} above and the following property:
  \begin{enumerate}[label=\Roman*)]
  \addtocounter{enumi}{3}
    \item In any prime field $\Fp$ with $q|p-1$, there exist $N$ distinct $q$th roots of unity $\boldsymbol{\rho} = (\rho_1, \dots, \rho_N)$, s.t.\begin{enumerate}
        \item $\bfV(\boldsymbol{\rho}, \boldsymbol{\gamma})$ is invertible and
        \item all $T\times T$ submatrices of $\bfV(\boldsymbol{\rho}, \as)$ and $\bfV(\boldsymbol{\rho}, \bs)$ are invertible,
      \end{enumerate}
  \end{enumerate}
  where %
  $\boldsymbol{\gamma} = \vect(\{\ap || \as\}+\{\bp || \bs\})$, with additions in $\Zq$. %
  \end{definition}
Note that Condition IV) is fulfilled whenever III) holds and   $\as$, $\bs$ are arithmetic progressions with common differences coprime to $q$, cf ~\cite[Lemma~1]{hofmeister2025cat}. 
  For $M=1$, \cref{def:gptables} recovers the definitions of OPP DTs/CATs.

\section{Grid Extension} \label{sec:extensions}

We define three extension operations, transforming 1) DTs for the OPP into DTs for the GP; 2) CATs for the OPP into CATs for the GP; and 3) DTs for the OPP into CATs for the GP, respectively.
The main idea is to construct an OPP DT/CAT for parameters $(K M, L, T)$, where $\ap$ is an arithmetic progression (or generalized arithmetic progression with chain length $r$) and extend it to a GP for $(K, M, L, T)$, by extending $\bp$ to $\bp \oplus \ap_{1:M} - \eap_1$. This results in a GP-DT/CAT where the blocks corresponding to $\TL_{k, \ell}$ have constant antidiagonals. As we will show later, if the original table is a valid DT/CAT, then the extended DT/CAT satisfies \cref{def:gptables} and thus corresponds to a valid PDMM scheme.

\begin{remark}
  The extension operations can be defined for DTs/CATs where $\ap$ is either an arithmetic progression or a generalized arithmetic progression\footnote{The format of our $\gap(\ell, x, r)$ definition is not required, i.e. $1$ does not need to be one of the common differences.} with chain length $M$. For ease of exposition we present the case where $\ap$ is an arithmetic progression here.
\end{remark}

  \begin{figure}[t]
  \vspace{1mm}
    \begin{subfigure}{0.33\linewidth}
      \centering
      \includegraphics[height=5.8cm]{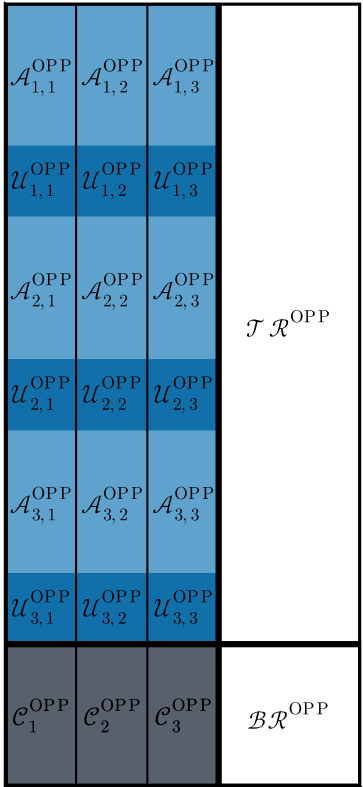}
      \caption{\small OPP.}
      \label{fig:illustrextOPP}
    \end{subfigure}
    \begin{subfigure}{0.65\linewidth}
      \centering
      \includegraphics[height=5.8cm]{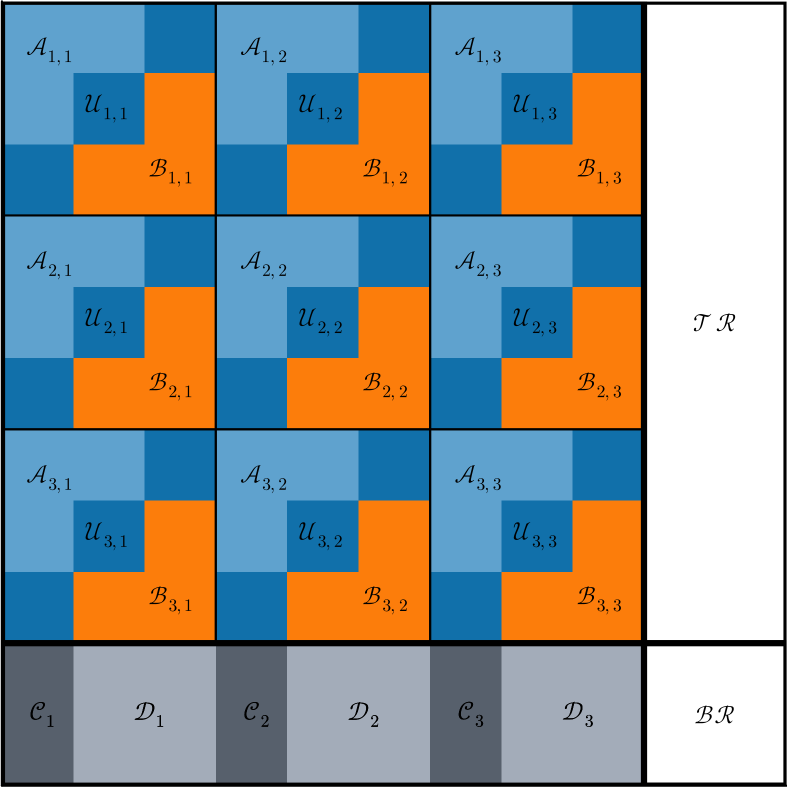}
      \caption{\small GP.}
      \label{fig:illustrextGP}
    \end{subfigure}
    \caption{ \small Illustration of the sets defined for an OPP-DT/CAT and a GP-DT/CAT for $K=M=L=3$ and $T=2$.}
\label{fig:illustrext}
\end{figure}

The formal definitions are as follows.

  \begin{definition}%
    \label{def:dttdt}
    The DT$\to$DT extension of a $(K M, L, T)$-OPP DT $(\ap, \bp, \as, \bs)$, where $\ap$ is an arithmetic progression, is defined as the $(K, M, L, T)$-GP degree table $(\apdtd, \bpdtd, \asdtd, \bsdtd)$ with 
    {\small
    \begin{equationarray*}{ll}
        \apdtd = \ap, &\bpdtd = \bp \oplus \ap_{1:M} - \eap_1,\\
        \asdtd= \as, & \bsdtd = \bs.
    \end{equationarray*}}
  \end{definition}

  \begin{definition}%
    \label{def:cattcat}
    The CAT$\to$CAT extension of a $(K M, L, T)$-OPP CAT $(\ap, \bp, \as, \bs, q^\prime)$, where $\ap$ is an arithmetic progression, is defined as the $(K,M, L, T)$ GP-CAT $(\apctc, \bpctc, \asctc, \bsctc, q=q^\prime)$ with 
    {\small \begin{equationarray*}{ll}
        \apctc = \ap, &\bpctc = \bp \oplus \ap_{1:M} - \eap_1 \bmod q \\
        \asctc= \as, & \bsctc = \bs.
    \end{equationarray*}}
  \end{definition}

  \begin{definition} %
    \label{def:dttcat}
    Let $(\ap, \bp, \as, \bs)$ be a $(K M, L, T)$-OPP DT, where $\ap$, $\as$, and $\bs$ are arithmetic progressions, $\eap_1=\ebp_1=0$ and the common difference of $\ap$ is $1$.
    The DT$\to$CAT extension of  $(\ap, \bp, \as, \bs)$  is defined as the $(K, M, L, T)$-GP-CAT $(\apdtc, \bpdtc, \asdtc, \bsdtc, q)$  with
    {\small\begin{equationarray*}{ll}
        \apdtc = \ap, &\bpdtc = \bp \oplus \ap_{1:M} - \eap_1 \bmod q \\
        \asdtc= \as, & \bsdtc = \bs,
    \end{equationarray*}}
    and $q = \max((\ap||\as) \oplus (\bp||\bs)) - M + 2 + \gamma$, 
    where $\gamma$ is the smallest non-negative integer s.t. $q$ is coprime to the common differences of $\as$ and $\bs$.
  \end{definition}

  \begin{theorem} %
    \label{thm:ext}
    The extension operations according to Definitions \ref{def:dttdt}, \ref{def:cattcat}, and \ref{def:dttcat}, result in valid DTs/CATs with $N \leq N^\prime + (M-1)(K+T)L$ workers where $N^\prime$ denotes the number of workers of the original OPP scheme. Further, $N\geq N^\prime$ for Definitions~\ref{def:dttdt} and \ref{def:cattcat} and $N\geq N^\prime-M+1$ for Definition \ref{def:dttcat}.
  \end{theorem}
  We provide an intuitive explanation here and defer the formal proof to 
  \ifarxiv
    \cref{sec:proofthmext}.
  \fi

  Consider the sets defined for the entries of an OPP table and a GP table in \cref{sec:prelims} and illustrated in \cref{fig:illustrext}.
  The DT$\to$DT and CAT$\to$CAT extension operations are designed such that the sets 
      $\cA_{k,\ell}$, $\cU_{k,\ell}$, $\cC_{\ell}$, $\TR$, $\BR$ in the GP table equal their counterpart $\cAo_{k,\ell}$, $\cUo_{k,\ell}$, $\cCo_{\ell}$, $\TRo$, $\BRo$ in the OPP table. In the DT$\to$CAT extension they are equal up to a modulo reduction.
  Further, the extension operations are designed in a way so as to make the $M\times M$ blocks in the top-left quadrant of the GP table have constant antidiagonals.
  By assumption, $\cAo_{k, \ell}$ forms an arithmetic progression which is then extended by one element by the singleton set $\cU_{k, \ell}$. The set $\cB_{k,\ell}$ then further extends this arithmetic progression by $M-1$ more elements.
  Because of this relation, if for some $k,k^\prime \in \setint{1}{K}$ and $\ell, \ell^\prime$, the set $\cB_{k,\ell} \cap\, \cU_{k^\prime, \ell^\prime} \neq \emptyset$, that would imply that $\cUo_{k,\ell}\cap \cAo_{k^\prime,\ell^\prime} \neq \emptyset$, i.e. the OPP scheme would be invalid.
  The sets $\cD_\ell$ relate to $\cC_\ell$ in a similar way.

  The upper bound on the numbers of workers follows from the fact that only the sets $\cB_{k, \ell}$ and $\cD_\ell$ are new in the GP compared to the OPP table. Further, $|\cB_{k, \ell}|\leq M-1$, since each block $\TL$ has constant antidiagonals. Thus, $|\bigcup_{k,\ell} \cB_{k,\ell}|\leq (M-1)KL$ and $|\bigcup_{\ell} \cD_{\ell}| \leq LT(M-1)$.

    Most known schemes for the OPP can be expressed as a degree table where $\ap$ is an arithmetic progression.
    For instance, \gaspr~\cite{doliveira2021degreec}, A3S~\cite{kakar2019capacity}, and \catx and \dogrs \cite{hofmeister2025cat}.
    As a consequence, a number of new schemes can be derived from these existing OPP schemes.
    DT$\to$CAT-extension can be applied to \gaspsmall, \gaspbig, A3S, and to \dogrs\footnote{Note that \dog$_{r=T,s=T}$ equals A3S.} whenever $r\in\{1, T\}$ and $s\in\{1, \min\{T,K+r\}\}$.
    DT$\to$DT-extension can be applied to all the schemes in the preceding list as well as \gaspr, and \dogrs for general $r$ and $s$.
    CAT$\to$CAT-extension can be applied to the \catx scheme.

    We evaluate the numbers of workers used by the new GP schemes arising from the extensions above in \cref{sec:comparison} and find that some compare favorably to the state-of-the-art for certain ranges of parameters.

    \begin{remark} \label{rem:addconstr}
    Applying the extension operations to OPP schemes produces GP-DTs/CATs that fulfill additional constraints (formally stated in \cref{lem:constraints}), not required for valid GP-DTs/CATs in general.
    This observation suggests that the number of workers used by a GP scheme can be reduced further by directly designing the scheme for the GP rather than extending an OPP scheme. 
    The only example, known to the authors, of a scheme that can be expressed as a GP-DT or GP-CAT and cannot be expressed as the extension of an OPP DT/CAT is the scheme of~\cite{byrne2023straggler}, which we refer to as the BGK scheme.
    \end{remark}

    \begin{lemma} \label{lem:constraints}
      For any $(K, M, L, T)$-GP table that can be expressed as the extension of a $(KM, L, T)$-OPP table according to Definition~\ref{def:dttdt},~\ref{def:cattcat}, or \ref{def:dttcat} the following hold for all $k,k^\prime \in \setint{1}{K}$ and $\ell, \ell^\prime \in \setint{1}{L}$ \begin{align*}
      &\cA_{k,\ell} \cap \cA_{k^\prime, \ell^\prime} \neq \emptyset \iff k=k^\prime, \ell=\ell^\prime \text{ and } \\
      &\cA_{k,\ell} \cap \TR = \cA_{k,\ell} \cap \BR = \emptyset.
      \end{align*}
    \end{lemma}
    \begin{proof}
      If either one of the conditions were not satisfied for the $(K, M, L, T)$-GP-DT/CAT, then in the $(KM, L, T)$-OPP DT/CAT, the set $\cAo_{k,\ell}$ intersects either $\cAo_{k^\prime, \ell^\prime}$ or $\TRo \cup \BRo$, making it invalid. 
    \end{proof}

    We design a GP-CAT not adhering to these constraints in \cref{sec:newscheme} and demonstrate that it requires fewer workers for a range of parameters in \cref{sec:comparison}.

\begin{remark}
  While exploring the extension operations given in \cref{sec:extensions}, we realized that some schemes from the literature can be related to each other via these extension operations. For instance, 
\ggasp \cite{karpuk2024modular} is the DT$\to$DT-extension of \gasp \cite{doliveira2020gaspa,doliveira2021degreec}. The ROU scheme \cite{machado2022root} is equivalent up to normalization and reordering to the DT$\to$CAT-extension of A3S~\cite{kakar2019capacity}. The DFT codes of \cite{mital2022secure} are the DT$\to$CAT-extension of \gaspbig, specialized for the IPP.
We will provide a formal proof and detailed explanations in an extended version of this paper.
\end{remark}

\section{A New Construction of GP-CATs} \label{sec:newscheme}

\begin{figure}
  \centering
    \resizebox{0.85\linewidth}{!}{
          \input{figures/rainbowscheme2_4_2_5}
        }
  \caption{\small The cyclic-addition degree table of \cref{constr:newscheme} for $K=L=2$, $M=4$ and $T=5$. For this case, $x=5$, $z=3$, $y=15$ and $q=29$.}
\end{figure}

We present a new construction of GP-CATs, that is not the extension of any OPP scheme.
To this end for given $K, M, L$ and $T$, define $\ztr \defeq L + \lceil \frac{K+T}{KM+K} \rceil$, $\zbl \defeq \lceil \frac{L+T-1}{K} \rceil$ and 
let $\zbr$ be defined as $\lfloor \frac{L+T-1}{KM-T+1}\rfloor + 1$ if $T\leq KM$ and $L+T-1 + \lfloor \frac{K+T}{KM+K} \rfloor$ otherwise.

    The construction is defined as follows.

\begin{construction} \label{constr:newscheme}

  Let $K\geq L$, $M\geq2$, and let \begin{equationarray*}{ccc}
    x = M+1, & y = z x, & q = Ky - 1, 
    \end{equationarray*}
    where $z = \max\{L+1, \ztr, \zbl, \zbr\}$. The construction is defined by $(\ap, \bp, \as, \bs, q)$, where
    {\small
  \begin{equationarray*}{cc}
    \ap \!= \!\gap(KM, \!y, \!M) \!\bmod q, & \bp \!=\! \gap(LM, \!x, \!M) \!\bmod q,  \\
    \as \! = \!x \!\cdot \! (0:T-1) \!- \!1 \!\bmod q,  &\!\!\! \bs \!= \! Lx \!+\! y \!\cdot\! (0:T-1) \!\bmod q.
    \end{equationarray*}
  }
\end{construction}
\begin{theorem} \label{thm:newscheme}
  \Cref{constr:newscheme} is a valid GP-CAT, i.e., it satisfies \cref{def:gptables} with $N\leq K(M+1)z-1$.
\end{theorem}
The proof of the theorem can be found in 
  \ifarxiv
    \cref{sec:proofthmext}.
  \fi

We confirm that this construction is indeed not the extension of a valid OPP scheme by \cref{lem:constraints}.
Note that $\ebp_2 + \eas_1 \bmod q = 0 = \ebp_1 + \eap_1$, meaning that $\cA_{1,1}\cap\TR \neq \emptyset$.

\section{Numerical Comparison} \label{sec:comparison}

In general we find that the scheme that performs best for a choice of parameters $K, M, L$ and $T$ differs widely.
To give a rough overview, we evaluate a multitude of schemes for all $2 \leq K, M, L, T \leq 20$ and record the number of instances where the respective scheme achieves the lowest number of workers in \cref{tab:scheme_comparison}. 
We have included all extensions listed in \cref{sec:extensions} that outperform all other schemes for at least one set of parameters. The GP scheme SGPD of \cite{aliasgari2020private} has not achieved the lowest number of workers for any set of parameters in this range and is not listed.
In the comparison, if multiple schemes achieve the same number of workers, the instance is counted for both in the second column of the table, but is also scored as $0\%$ difference to the next best scheme in the third column. For this reason, the second column is not expected to sum to $19^4$. With the exception of the MP codes of~\cite{karpuk2024modular}, for which we use the theoretical lower bound on the number of workers, we use the exact number of workers for each scheme. The chain lengths in \gaspr, \ggaspr, and \dogrs were optimized for each instance.

We observe that the extensions of \catx and \dogrs improve upon the state-of-the-art in many instances. 
BGK and the construction of \cref{sec:newscheme} stand out both in the number of instances where they achieve the lowest number of workers and in the gap to the next best scheme.
This reinforces the intuition of \cref{rem:addconstr}, that constructions that do not correspond to extensions of valid OPP DTs/CATs have the potential to achieve low numbers of workers.

\begin{table}[t]
\vspace{4mm}
    \caption{\small Comparison of the number of workers for all $2 \leq K, M, L, T \leq 20$. The second column lists the number of instances where the respective scheme achieves the lowest number of workers (out of the $19^4=130321$ instances). The third and fourth columns report the average and maximum relative difference to the next best scheme. The table is divided into three sections: the first lists GP schemes from the literature, the second lists extensions applied to OPP schemes, and the third consists of our scheme of \cref{sec:newscheme}.}
\vspace{-2mm}
    \label{tab:scheme_comparison}
    \centering
    \resizebox{0.98\linewidth}{!}{
    \begin{tabular}{l c l l}
        \toprule
        Scheme & \# Best & Avg. Diff. & Max. Diff. \\
        \midrule
          BGK~\cite{byrne2023straggler} & 27950 & 7.63\% & 27.5\% \\
          ROU~\cite{machado2022root} & 7956 & 1.17\% & 13.9\% \\
          MP~\cite{karpuk2024modular} & 1354 & 1.19\% & 7.7\% \\
          \ggaspr~\cite{karpuk2024modular} & 19 & 0.06\% & 0.5\% \\
        \midrule
          DT$\to$DT: \dogrs \cite{hofmeister2025cat} & 38249 & 2.31\% & 17.4\% \\
          CAT$\to$CAT: \catx \cite{hofmeister2025cat} & 27053 & 0.72\% & 8.3\% \\
          DT$\to$CAT: DOG$_{1,\min\{T,K+1\}}$ \cite{hofmeister2025cat} & 1594 & 0.33\% & 2.9\% \\
        \midrule
          Construction in \cref{sec:newscheme} & 29414 & 6.80\% & 31.0\% \\
        \bottomrule
    \end{tabular}
    }
\end{table}

\section{Conclusion}
We have shown ways to extend a class of OPP schemes to the GP and have found that in many cases the resulting GP schemes use fewer workers than existing approaches.
Despite this, we have shown that GP schemes that are the extension of an OPP scheme fulfill certain additional constraints, which are at odds with minimizing the number of workers.
A more detailed treatment of how some existing GP schemes can be expressed as extensions of OPP schemes and additional constructions directly for the GP are left for future work.

\clearpage
\bibliographystyle{IEEEtran}
\bibliography{references}

\begin{thebibliography}{10}
\providecommand{\url}[1]{#1}
\csname url@samestyle\endcsname
\providecommand{\newblock}{\relax}
\providecommand{\bibinfo}[2]{#2}
\providecommand{\BIBentrySTDinterwordspacing}{\spaceskip=0pt\relax}
\providecommand{\BIBentryALTinterwordstretchfactor}{4}
\providecommand{\BIBentryALTinterwordspacing}{\spaceskip=\fontdimen2\font plus
\BIBentryALTinterwordstretchfactor\fontdimen3\font minus \fontdimen4\font\relax}
\providecommand{\BIBforeignlanguage}[2]{{%
\expandafter\ifx\csname l@#1\endcsname\relax
\typeout{** WARNING: IEEEtran.bst: No hyphenation pattern has been}%
\typeout{** loaded for the language `#1'. Using the pattern for}%
\typeout{** the default language instead.}%
\else
\language=\csname l@#1\endcsname
\fi
#2}}
\providecommand{\BIBdecl}{\relax}
\BIBdecl

\bibitem{arxivversion}
C.~Hofmeister, R.~Tajeddine, A.~Wachter-Zeh, and R.~Bitar, ``On the extension of private distributed matrix multiplication schemes to the grid partition,'' \emph{arXiv preprint}, Jan. 2026.

\bibitem{chang2018capacity}
W.-T. Chang and R.~Tandon, ``On the {{Capacity}} of {{Secure Distributed Matrix Multiplication}},'' in \emph{2018 {{IEEE Global Communications Conference}} ({{GLOBECOM}})}, Dec. 2018, pp. 1--6.

\bibitem{chang2019upload}
------, ``On the {{Upload}} versus {{Download Cost}} for {{Secure}} and {{Private Matrix Multiplication}},'' in \emph{2019 {{IEEE Information Theory Workshop}} ({{ITW}})}, Aug. 2019, pp. 1--5.

\bibitem{doliveira2021degreec}
R.~G.~L. D'Oliveira, S.~El~Rouayheb, D.~Heinlein, and D.~Karpuk, ``Degree {{Tables}} for {{Secure Distributed Matrix Multiplication}},'' \emph{IEEE Journal on Selected Areas in Information Theory}, vol.~2, no.~3, pp. 907--918, Sep. 2021.

\bibitem{makkonen2023algebraica}
O.~Makkonen, E.~Sa{\c c}{\i}kara, and C.~Hollanti, ``Algebraic {{Geometry Codes}} for {{Secure Distributed Matrix Multiplication}},'' \emph{arXiv preprint arXiv:2303.15429}, Jun. 2023.

\bibitem{doliveira2020gaspa}
R.~G.~L. D'Oliveira, S.~El~Rouayheb, and D.~Karpuk, ``{{GASP Codes}} for {{Secure Distributed Matrix Multiplication}},'' \emph{IEEE Transactions on Information Theory}, vol.~66, no.~7, pp. 4038--4050, Jul. 2020.

\bibitem{aliasgari2020private}
M.~Aliasgari, O.~Simeone, and J.~Kliewer, ``Private and {{Secure Distributed Matrix Multiplication With Flexible Communication Load}},'' \emph{IEEE Trans. Inform. Forensic Secur.}, vol.~15, pp. 2722--2734, 2020.

\bibitem{makkonen2024general}
O.~Makkonen and C.~Hollanti, ``{General Framework} for {Linear Secure Distributed Matrix Multiplication} with {Byzantine Servers},'' \emph{IEEE Transactions on Information Theory}, vol.~70, no.~6, pp. 3864--3877, 2024.

\bibitem{karpuk2024modular}
D.~Karpuk and R.~Tajeddine, ``Modular {{Polynomial Codes}} for {{Secure}} and {{Robust Distributed Matrix Multiplication}},'' \emph{IEEE Transactions on Information Theory}, vol.~70, no.~6, pp. 4396--4413, 2024.

\bibitem{machado2022root}
R.~A. Machado and F.~Manganiello, ``Root of {{Unity}} for {{Secure Distributed Matrix Multiplication}}: {{Grid Partition Case}},'' in \emph{2022 {{IEEE Information Theory Workshop}} ({{ITW}})}, Nov. 2022, pp. 155--159.

\bibitem{byrne2023straggler}
E.~Byrne, O.~W. Gnilke, and J.~Kliewer, ``Straggler- and {{Adversary-Tolerant Secure Distributed Matrix Multiplication Using Polynomial Codes}},'' \emph{Entropy}, vol.~25, no.~2, p. 266, Feb. 2023.

\bibitem{hofmeister2025cat}
C.~Hofmeister, R.~Bitar, and A.~Wachter-Zeh, ``Cat and dog: Improved codes for private distributed matrix multiplication,'' in \emph{2025 IEEE International Symposium on Information Theory (ISIT)}, 2025, pp. 1--6.

\bibitem{lopez2022secure}
H.~H. Lopez, G.~L. Matthews, and D.~Valvo, ``Secure {{MatDot}} codes: A secure, distributed matrix multiplication scheme,'' in \emph{2022 {{IEEE Information Theory Workshop}} ({{ITW}})}, Nov. 2022, pp. 149--154.

\bibitem{mital2022secure}
N.~Mital, C.~Ling, and D.~G{\"u}nd{\"u}z, ``Secure {{Distributed Matrix Computation With Discrete Fourier Transform}},'' \emph{IEEE Transactions on Information Theory}, vol.~68, no.~7, pp. 4666--4680, Jul. 2022.

\bibitem{machado2023hera}
R.~A. Machado, G.~L. Matthews, and W.~Santos, ``{{HerA Scheme}}: {{Secure Distributed Matrix Multiplication}} via {{Hermitian Codes}},'' in \emph{2023 {{IEEE International Symposium}} on {{Information Theory}} ({{ISIT}})}, Jun. 2023, pp. 1729--1734.

\bibitem{makkonen2024flexible}
O.~Makkonen, ``{Flexible Field Sizes} in {Secure Distributed Matrix Multiplication} via {Efficient Interference Cancellation},'' \emph{arXiv preprint arXiv:2404.15080}, 2024.

\bibitem{kakar2019capacity}
J.~Kakar, S.~Ebadifar, and A.~Sezgin, ``On the {Capacity} and {Straggler}-{Robustness} of {Distributed} {Secure} {Matrix} {Multiplication},'' \emph{IEEE Access}, vol.~7, pp. 45\,783--45\,799, 2019.

\end{thebibliography}

\ifarxiv
\clearpage
\appendices
\crefalias{section}{appendix}

\setcounter{theorem}{0}
\section{Proof of \cref{thm:ext}}  \label{sec:proofthmext}

For convenience we restate the theorem.
  \begin{theorem} %
    
    The extension operations according to Definitions \ref{def:dttdt}, \ref{def:cattcat}, and \ref{def:dttcat}, result in valid DTs/CATs with $N \leq N^\prime + (M-1)(K+T)L$ workers where $N^\prime$ denotes the number of workers of the original OPP scheme. Further, $N\geq N^\prime$ for Definitions~\ref{def:dttdt} and \ref{def:cattcat} and $N\geq N^\prime-M+1$ for Definition \ref{def:dttcat}.
  \end{theorem}

  \begin{proof}
    We prove the validity of all three types of extensions DT$\to$DT, CAT$\to$CAT, and DT$\to$CAT simultaneously, specializing our arguments where necessary. The $\bmod~q$ operations vanish throughout the proof for the DT$\to$DT extension.

    Let $(\ape, \bpe, \ase, \bse, q^\prime = q)$ be the extension of a $(K\cdot M, L, T)$-OPP DT $(\ap, \bp, \as, \bs)$ or $(K\cdot M, L, T)$-OPP CAT $(\ap, \bp, \as, \bs, q)$.
    Assume $\ap$ is an arithmetic progression and when $E=C\to C$ or $E=D\to C$, additionally assume that $\as$ and $\bs$ are arithmetic progressions. Denote the common difference of $\ap$ as $\ca$.
    In general we show that the GP-DT/CAT resulting from extension fulfills \cref{def:gptables} based on the fact that the OPP-DT/CAT fulfills \cref{def:gptables}.

    By design, the extension ensures that the blocks of the top-left quadrant of the GP-DT/CAT have constant antidiagonals. Specifically, for $k\in \setint{1}{K}, \ell \in \setint{1}{L}$ consider the $(k,\ell)$th block $\TL_{k, \ell}$. Define $\kb=k-1$ and $\lb=\ell-1$. For $i,j\in\setint{1}{M}$, the $(i,j)$th entry of the block is given by\begin{align*}
      \ap_k[i] + \bp_\ell[j] &= (\eap_1 + (\kb M+i-1) \ca) + (\ebp_\ell \\&\qquad + (j-1)\ca) \bmod q \\
                             &= \eap_1 + \ebp_\ell + (\kb M +i+j-2) \ca \bmod q,
    \end{align*}
    and does not depend on $i$ and $j$ individually, but rather on $i+j$.
    Note that $\cA_{k,\ell} = \cAo_{k,\ell} = \{\eap_1 + \ebp_\ell + (\kb M +\mu) \ca \bmod q  \st \mu \in \setint{0}{M-2}\}$ and $\cB_{k,\ell} = \{\eap_1 + \ebp_\ell + (\kb M + M +\mu) \ca \bmod q  \st \mu \in \setint{0}{M-2}\}$.
    Further, note that the sets $\cU_{k, \ell}$ are singletons and equal $\cUo_{k, \ell}$.
    We denote the single element by \begin{align*}
      u_{k, \ell} &\defeq \eap_1 + \kb M \ca + \ebp_\ell + (M-1) \ca \bmod q \\
                  &= \eap_1 + \ebp_\ell + (kM-1) \ca \bmod q \\
                  &= \eap_k + \ebp_\ell + (M-1) \ca \bmod q.
    \end{align*}

    Note that $\cC_{\ell} = \{\as\} + \ebp_\ell \bmod q = \cCo$ and $\cD_{\ell} = \{\as\} + \{\ebp_\ell + i \ca \st i \in \setint{1}{M-1}\} \bmod q$.

    For the DT$\to$CAT extension, note the choice of $q$, together with the additional assumptions $\eap_1=\ebp_1=0$ and $\eap_2=1$, which ensure that any time the modulo reduction causes a wrap around, the result falls between $0$ and $M-1$, i.e. in $\cA_{1,1}$.

    We confirm the properties in \cref{def:gptables} one-by-one in the order II) a)-e), III), IV), I).

    Let $k, k^\prime in \setint{1}{K}, \ell, \ell^\prime \in \setint{1}{L}$.
    \emph{II) a) of the GP-DT/CAT} requires that $u_{k,\ell} \neq u_{k^\prime, \ell^\prime}$ unless $k=k^\prime$ and $\ell=\ell^\prime$. If this were not the case, then we would have $\cUo_{k, \ell}=\cUo_{k^\prime, \ell^\prime}$, violating property II) a) for the OPP-DT/CAT.

    \emph{II) b) of the GP-DT/CAT} requires that the $u_{k,\ell}$ are not elements of $\TR$. Note that $\TRo = \TR$ and that by property II) b) of the OPP-GP-CAT $\cUo_{k, \ell} = \{u_{k,\ell}\}$ and $\TRo$ are disjoint.

    \emph{II) c) of the GP-DT/CAT} requires that the $u_{k,\ell}$ are not elements of $\BL$. We first show that that they are not elements of $\cC_{\ell^\prime}$ for any $\ell^\prime \in \setint{1}{L}$. Assume the opposite then $u_{k, \ell}\in \cCo_{\ell^\prime}$, contradicting property II) c) of the OPP-DT/CAT. It remains to show that $u_{k, \ell} \notin \cD_{\ell^\prime}$ for any $\ell^\prime in \setint{1}{L}$. Assume the opposite, then  for some $k, \ell, \ell^\prime$ and $i\in \setint{1}{M-1}$, $j\in \setint{1}{T-1}$ it holds that
    \begin{align*}
      \ap_k + \ebp_\ell\!+ (M-1) \ca \bmod q &= \eas_j + \ebp_{\ell^\prime}\!+ i \ca \bmod q \\
      \ap_k + \ebp_\ell + (M-1 - i) \ca \bmod q &= \eas_j + \ebp_{\ell^\prime} \bmod q.
    \end{align*}
    However, note that the left hand side is an element of $\cAo_{k, \ell}$ and the right hand side is an element of $\cCo_{\ell^\prime}$.

    \emph{II) d) of the GP-DT/CAT} requires that the $u_{k,\ell}$ are not elements of $\BR$. Note that $\cUo_{k,\ell}$ and $\BRo$ are disjoint.

    \emph{II) e) of the GP-DT/CAT} requires that the $u_{k,\ell}$ are not elements of non-main-antidiagonal entries in $\TR$, i.e. not elements of $\cA_{k^\prime,\ell^\prime}$ or $\cB_{k^\prime,\ell^\prime}$ for any $k^\prime\in \setint{1}{K}, \ell^\prime \in \setint{1}{L}$. Since the $\cAo_{k^\prime,\ell^\prime}$ are disjoint of the $\cUo_{k, \ell}$, we focus on $\cB_{k^\prime,\ell^\prime}$. Assume $u_{k,\ell}\in\cB_{k^\prime,\ell^\prime}$, then there exists a $\mu \in \setint{0}{M-2}$ for which 
    \begin{align*}
      &\eap_k + \ebp_\ell + (M-1) \ca \bmod q \\&\qquad= \eap_{k^\prime} + \ebp_{\ell^\prime} + (M +\mu) \ca \bmod q   \\
      &\eap_k + \ebp_\ell + (M-\mu-2) \ca \bmod q \\&\qquad= \eap_{k^\prime} + \ebp_{\ell^\prime} + (M-1) \ca \bmod q \\
      &\eap_k + \ebp_\ell + (M-\mu-2) \ca \bmod q \\&\qquad= u_{k^\prime, \ell^\prime}.
    \end{align*}
    However, the left hand side is in $\cAo_{k, \ell}$.

    \emph{III) of the GP-DT/CAT} requires that the $\ape||\ase$ are distinct and the $\bpe||\bse$ are distinct. For the $\ape||\ase$ it follows directly from III) of the OPP-DT or IV) of the OPP-GP.
    For $\bpe||\bse$ we have that \begin{align*}
      \bpe &= \bp \oplus \ap_{1:M} - \eap_1 \bmod q \\
               &= \bp \oplus \ca\vecint{0}{M-1}\bmod q.
    \end{align*}

    Assume all entries of $\bp||\bs$ but not of $(\bp \oplus \ap_{1:M} - \eap_1 \bmod q)||\bs$ are pairwise distinct.
    There must either a) exist indices $\ell,\ell^\prime \in\setint{1}{L}$ and $m, m^\prime \in \setint{0}{M-1}$ s.t. $\ebp_\ell + \ca m \bmod q= \ebp_{\ell^\prime} + \ca m^\prime \bmod q$; or b) exist indices $\ell\in\setint{1}{L}, m\in\setint{0}{M-1}$ and $t\in\setint{1}{T}$ s.t. $\ebp_\ell + \ca m = \ebs_{t} \bmod q$. 
    Condition a) cannot hold since it implies $\ebp_\ell + \ca m + \eap_1 \bmod q = \ebp_{\ell^\prime} + \ca m^\prime + \eap_1 \bmod q$, i.e. $\ebp_\ell + \eap_{1+m} \bmod q = \ebp_{\ell^\prime} + \eap_{1+m^\prime} \bmod q$, contradicting II) a) of the OPP-DT/CAT. 
    Condition b) implies \begin{align*}
      \ebp_\ell + \ca m + \eap_1 \bmod q &= \ebs_{t} + \eap_1 \bmod q \\
      \ebp_\ell + \eap_{1+m} \bmod q &= \ebs_{t} + \eap_1 \bmod q \\
    \end{align*}
    However, the left hand side is in the top-left quadrant of the OPP-DT/CAT and the right hand side is in the top-right quadrant of the OPP-DT/CAT, contradicting II) b) of the OPP-DT/CAT.

    \emph{IV) of the GP-CAT} holds for the DT$\to$CAT and CAT$\to$CAT extensions since III) holds and $\as$ and $\bs$ are arithmetic progressions with common differences coprime to $q$.

    \emph{I) of the GP-CAT} consists of an upper and lower bound on the number of distinct entries in the table. 
    The lower bound holds since any element of the OPP-DT/CAT is in $\cAo_{k,\ell}$, $\cUo_{k,\ell}$, $\cCo_\ell$, $\TRo$, or $\BRo$ and thus also in $\cA_{k,\ell}$, $\cU_{k,\ell}$, $\cC_\ell$, $\TR$, or $\BR$, for some $k\in\setint{1}{K}, \ell \in \setint{1}{L}$. 
    For the DT$\to$CAT extension, the modulo operation can reduce the number of distinct entries in the table by at most $M-1$.
    The upper bound holds since the additional sets $\cB_{k, \ell}$ consist of at most $M-1$ elements each and the sets $\cD_\ell$ of at most $T(M-1)$, each.
  \end{proof}

\section{Proof of \cref{thm:newscheme}}  \label{sec:proofthmnewscheme}

For convenience we restate the theorem.
\begin{theorem} %
  \Cref{constr:newscheme} is a valid GP-CAT, i.e., it satisfies \cref{def:gptables} with $N\leq K(M+1)z-1$.
\end{theorem}

  We introduce and prove the following lemma before proving the theorem.
  \begin{lemma} \label{lem:uklmodetc}
    Let $K, M, L, x, y$ be like in \cref{constr:newscheme} and let $k\in\setint{1}{K}$, $\ell \in \setint{1}{L}$, $\mu \in \setint{0}{2M-2}$. Define $\kb=k-1$ and $\lb=\ell-1$ and let $u_{k, \ell} = \kb y + \lb x + M-1 \bmod q$. 
    Then the following hold:\begin{enumerate}
      \item $0 \leq \kb y + \lb x + \mu < q-1,$
      \item $u_{k,\ell} = \kb y + \lb x + M-1,$
      \item $u_{k,\ell} \bmod x = M-1,$
      \item $u_{k,\ell} \bmod y = \lb + M-1.$
    \end{enumerate}
  \end{lemma}
  \begin{proof}
    We start by establishing statement 1). 
    Since $\kb, \lb, \mu, x, y \geq 0$, we have $\kb y + \lb x +\mu\geq 0$. Further, \begin{align*}
      \kb y + \lb x + \mu &\leq (K-1)y + (L-1)x + 2M-2 \\
                          &= Ky - 1 -y + Lx -x + 2M-1 \\
                          &= q -y + Lx +x-3 \\
                          &= q -x(z-L-1) - 3 \\
                          &<q-1,
    \end{align*}
    where we use $z\geq L+1$.
    
    Statement 2) follows from statement 1) with $\mu=M-1$.

    Now we prove statement 3). By statement 2) we have that\begin{align*}
      u_{k,\ell} \bmod x &= \kb y + \lb x + M-1 \bmod x \\
                         &= (\kb z + \lb) x + M-1 \bmod x \\
                         &= M-1.
    \end{align*}

    Finally, to prove statement 4), note that $\lb+M-1\leq L+M-2 \leq LM \leq (M+1)z = y$. With this we have
    \begin{align*}
      u_{k,\ell} \bmod y &= \kb y + \lb x + M-1 \bmod y \\
                         &= \lb x + M-1 \bmod zx \\
                         &= (\lb \bmod z) + M-1 \bmod zx \\
                         &= \lb + M-1.
    \end{align*}
  \end{proof}

\begin{proof}[Proof of \cref{thm:newscheme}]
  Let $k, k^\prime \in \setint{1}{K}$, $\ell, \ell^\prime \in \setint{1}{L}$, $\mu \in \setint{0}{2M-2}$, $m\in \setint{0}{M-1}$, $t, t_1, t_2 \in \setint{0}{T-1}$ with $k\neq k^\prime$ and $\ell \neq \ell^\prime$.
  For convenience of notation let $\kb \defeq k-1$, $\kbp \defeq k^\prime-1$, $\lb \defeq \ell-1$, and $\lbp \defeq \ell^\prime-1$.
  For all $k$ and $\ell$ we have that \begin{align*}
    \TL_{k,\ell} &= \{\ap_{k} \oplus \bp_{\ell} \bmod q\} \\
                 &= \{\kb y + \lb x + \mu \bmod q \st \mu \in \setint{0}{2M-2} \}, \\ 
                 &= \{\kb y + \lb x + \mu \st \mu \in \setint{0}{2M-2} \},
  \end{align*} where the modulo operation vanishes by statement 1) of \cref{lem:uklmodetc}.  Similarly, $\cU_{k, \ell} = \{\kb y+\lb x + M-1\}$ and $\TLt_{k, \ell} = \{\kb y + \lb x + \mu \st \mu \in \setint{0}{2M-2}, \mu \neq M-1 \}$. Define $u_{k, \ell} =\kb y+\lb x + M-1$, to denote the singular element of $U_{k, \ell}$.

  We establish that the conditions in \cref{def:gptables} hold one-by-one.

  \emph{Condition I)} is trivially satisfied.

  \emph{Condition II)} requires that the main antidiagonals of distinct blocks are distinct, i.e., that $u_{k, \ell} \neq u_{k^\prime, \ell^\prime}$. This is true since the $u_{k, \ell}$ are strictly increasing in left-to-right, top-to-bottom order. 
  Formally, assume the opposite and consider\begin{align*}
    u_{k, \ell} &= u_{k^\prime, \ell^\prime} \\ 
    \kb y+\lb x + M-1 &= \kbp y+\lbp x + M-1 \\
    (\kb-\kbp) y &= (\lbp-\lb) x  \\
    (\kb-\kbp) z &= \lbp-\lb,
  \end{align*}
which is impossible as $z\geq L+1 > |\lbp-\lb|$.

  \emph{Condition III) a)} requires that no entry $u_{k, \ell}$ of an antidiagonal appears in the top-right quadrant.
  We have that\begin{align*}
    \TR &= \ap \oplus \bs \bmod q  \\
        &= \{\kbp y + m + Lx + t y \bmod q \st \kbp \in \setint{0}{K-1}, \\
        &\qquad t\in\setint{0}{T-1}, m \in \setint{0}{M-1}\}.
  \end{align*}

  For ease of notation and for given $\kbp, t$ and $m$ let $\tmpvar \defeq Lx + \kbp y + t y + m$. 
  We prove the statement by showing that $\tmpvar \bmod q$ and $u_{k, \ell}$ have different residues modulo $y$. By statement 4) of \cref{lem:uklmodetc} we have that $(u_{k, \ell} \bmod y) \in \{M-1, \dots, L+M-2\}$. 

    It remains to show that the residue of $\tmpvar \bmod q$ is outside this range. It holds that \begin{align*}  
      \tmpvar \bmod q \bmod y &= \tmpvar  - \lfloor \frac{\tmpvar}{q} \rfloor q \bmod y \\
                              &= Lx\!+\!y t\!+\!k^\prime\!y\!+ m\!-\!\lfloor \frac{\tmpvar}{q} \rfloor (Ky-1) \bmod y \\
                              &= Lx + m  + \lfloor \frac{\tmpvar}{q} \rfloor \bmod y.
    \end{align*}
    Note that $Lx + m  + \lfloor \frac{\tmpvar}{q} \rfloor \geq Lx > L+M-2$. It remains to show that $\tmpvar \bmod q \bmod y < M-1$, which we will do by demonstrating that $Lx + m  + \lfloor \frac{\tmpvar}{q} \rfloor < y+M-1$.

    The condition is equivalent to\begin{align*}
      y+M-1 &> Lx + M-1  + \lfloor \frac{\tmpvar}{q} \rfloor \\
      xz &> Lx + \lfloor \frac{\tmpvar}{q} \rfloor \\
      z &> L + \frac{\lfloor \frac{\tmpvar}{q} \rfloor}{x}. 
    \end{align*}

    And is fulfilled whenever $qx(z-L)>\tmpvar$.

    Using $z\geq \ztr$, we bound the left hand side from below as
    \begin{align*}
      qx(z-L) &\geq (Kxz-1)x\frac{K+T}{Kx} \\
              &= (Kxz-1)\frac{K+T}{K} \\
              &= (Kxz-1)(1+\frac{T}{K}) \\
              &= (K+T)xz-\frac{T}{K} - 1.
    \end{align*}

    We have that
    \begin{align*}
      qx(z-L) - \tmpvar &\geq (K+T)xz-\frac{T}{K} - 1 \\
                        &\qquad - (K+T-2)xz-Lx-M+1 \\
                        &= 2xz-\frac{T}{K} -(L+1)x+1 \\
                        &\geq 2xL + 2\frac{K+T}{K}-\frac{T}{K} -(L+1)x+1 \\
                        &= xL + \frac{T}{K}-x+3 \\
                        &= x(L-1) + \frac{T}{K}+3 \\
                        &\geq 0.
    \end{align*}
    Thus, $\tmpvar\bmod q\bmod y > L+M-2$ or $\tmpvar\bmod q\bmod y < M-1$ while $M-1\leq u_{k, \ell} \bmod y \leq L+M-2$.

  \emph{Condition III) b)} requires that no entry of a main antidiagonal $u_{k, \ell}$ appears in the bottom-left quadrant.
  We have that\begin{align*}
    \BL &= \{\as \oplus \bp \bmod q\} \\
        &= \{tx + \lbp x + m -1 \bmod q \st \lbp \in \setint{0}{L-1},  \\ 
        &\qquad t \in \setint{0}{T-1}, m \in \setint{0}{M-1} \}.
  \end{align*}

  Let $\tmpvarbl = tx + \lbp x + m -1$ and consider $u_{k, \ell} = \kb y + \lb x + M-1$. 
  When $\lbp=t=m=0$, then $\tmpvarbl \bmod q= q-1$ which exceeds $u_{k,\ell}$ by statement 1) of \cref{lem:uklmodetc}.

  For the case where $\lbp, t$, or $m$ exceeds $0$, we consider residues modulo $x$ to establish disjointness.

  In this case we have $\tmpvarbl >0$ and $\tmpvarbl < q$ since\begin{align*} 
    \tmpvarbl - q &= tx + \lbp x + m -1 - (Ky-1) \\
           &\leq (L+T-2)x + M -Ky - 1 \\
           &= (L+T-1-Kz)x -2 \\
           &\leq (L+T-1-K\zbl)x -2 \\
           &\leq (L+T-1 - K \lceil \frac{T+L-1}{K} \rceil )x -2 \\
           &< -2.
  \end{align*}

  Hence, we have \begin{align*}
    \tmpvarbl \bmod q \bmod x &= \tmpvarbl \bmod x \\
                           &= m-1 \bmod M+1 \\ 
                           &\in \setint{0}{M-2} \cup \{M\}
  \end{align*}

  However, $u_{k, \ell} \bmod x = M-1$ by statement 3) of \cref{lem:uklmodetc}.

  \emph{Condition III) c)} requires that no entry of an antidiagonal $u_{k, \ell}$ appears in the bottom-right quadrant.

  We have that \begin{align*}
    \BR &= \{\as \oplus \bs \bmod q \} \\
        &= \{Lx-1 + x t_1 + yt_2 \bmod q \vert t_1, t_2 \in \setint{0}{T-1}\}
  \end{align*}
  Note that the unique inverse of $x$ modulo $q$ is $Kz$.
    Let $\tmpvarbr = Lx-1 + x t_1 + yt_2$ and consider $(\tmpvarbr-M+1)x^{-1} \bmod q$.  
    We have \begin{align*}
      (\tmpvarbr\!-\!M\!+\!1)x^{-1} \bmod q &= (Lx-1 + x t_1 + y t_2 \\
                                            &\qquad -M+1)x^{-1} \bmod q \\
                                     &= (L + t_1 + z t_2) -Mx^{-1} \bmod q \\
                                     &= t_1 + z t_2 +L -M Kz + q \bmod q \\
                                     &= t_1 + z t_2 +L -M Kz \\&\qquad+ K(M+1)z - 1 \bmod q \\
                                     &= t_1 + z t_2 + L + Kz -1 \bmod q.
    \end{align*}
 
    We proceed by case distinction on whether $KM\geq T$.

    \emph{Condition III) c), Case $KM\geq T$.}

    We begin by establishing that $0 < t_1 + z t_2 + L + Kz -1 < q$, i.e. that the modulo operation vanishes.
    The lower bound holds since $L+Kz>1$ as $K\geq2$ and $z>L\geq2$. 
    The upper bound holds since $z\geq \zbr$ as the upper bound is equivalent to
    \begin{align*}
      (T-1) + z (T-1) + L + Kz -1 &< Kxz-1 \\
      (T-1) + z (T-1) + L + Kz &< Kxz \\
      (T-1) + z (T-1) + L &< KMz \\
      (T-1) + L &< z(KM-T+1) \\
      \frac{T+L-1}{KM-T+1} &< z.
    \end{align*}

    Thus, we have \begin{align*}
      (\tmpvarbr -M+1)x^{-1} \bmod q &= t_1 + z t_2 + L + Kz -1 \bmod q \\
                                     &= t_1 + z t_2 + L + Kz -1 \\
                                     &\geq L + Kz -1.
      \end{align*}
      Applying the same transformation to $u_{k, \ell}$ gives \begin{align*}
        (u_{k, \ell}-M+1) x^{-1} \bmod q &= (kx z + \ell x) x^{-1} \bmod q  \\
                                         &= k z + \ell  \bmod q  \\
                                         &= k z + \ell \\
                                         &\leq (K-1) z + (L-1) \\
                                         &\leq Kz + L-z-1 \\
                                         &\leq Kz -1. 
      \end{align*}
      Thus, after translating by $-M+1$ and scaling by $x^{-1}$, the value of $u_{k, \ell}$ is smaller than any value in $\BR$.

    \emph{Condition III) c), Case $KM<T$.}

    Define $t_3 = t_1 + zt_2$. Since $z\geq \zbr > T$, we have that $\{t_3 \st t_3 \in \gap(T^2, z, T)\} = \{t_1 + zt_2 \st t_1, t_2 \in \setint{0}{T-1}\}$. 
    We have\begin{align*}
      (\eta -M+1)x^{-1}\bmod q &= t_1 + zt_2 + L + Kz - 1 \bmod q \\
                               &= t_3 + L + Kz - 1 \bmod q.
    \end{align*}

    The modulo operation vanishes if $t_3 + L + Kz - 1<q$, i.e. when  
    $t_3<KMz-L$. Equivalently, it vanishes for all $t_3 \in \gap(KMT,z,T)$ and does not for all $t_3 \in KMz + \gap(T(T-KM), z, T)$.

    The remaining proof for the case $t_3 \in \gap(KMT,z,T)$ works by the same steps as the previous case $KM<T$ since now the modulo operation also vanishes and $\zbr$ is at least as big as in the previous case.
    Specifically, $L+Kz -1 \leq L + Kz - 1 + \gap(KMT, z, T) < q$.

    Thus, it remains to complete the proof for $t_3 \in KMz + \gap(T(T-KM), z, T)$. Let $t^\prime_3= t_3-KMz$, s.t. $t^\prime_3 \in \gap(T(T-KM), z, T)$ and $t^\prime_2 \in \setint{0}{T-KM-1}$.

    We have\begin{align*}
      (\eta -M+1)x^{-1}\bmod q &= t^\prime_3 + KMz + L + Kz - 1 \bmod q \\
                               &= t^\prime_3 + L + (Kxz - 1) \bmod q \\
                               &= t^\prime_3 + L \bmod q\\
                               &= t_1 + t^\prime_2z + L \bmod q.
    \end{align*}

    We apply long division by $Kx$ to $t^\prime_2$ and use $Kxz\bmod q=1$ to obtain \begin{align*}
      &(\eta -M+1)x^{-1}\bmod q  \\
      &= t_1 + \lfloor\frac{t^\prime_2}{Kx}\rfloor Kxz + (t^\prime_2 \bmod Kx)z + L \bmod q  \\
                               &= t_1 + \lfloor\frac{t^\prime_2}{Kx}\rfloor + (t^\prime_2 \bmod Kx)z + L \bmod q.
    \end{align*}
    The modulo operation vanishes as \begin{align*}
          &t_1 + \lfloor\frac{t^\prime_2}{Kx}\rfloor + (t^\prime_2 \bmod Kx)z + L - q  \\
          &\leq L+T + \lfloor\frac{T-KM-1}{Kx}\rfloor + (T-KM-1 \bmod Kx)z \\&\qquad- Kxz \\
          &\leq L+T + \lfloor\frac{T-KM-1}{Kx}\rfloor + ((T-KM-1 \bmod Kx)\\&\qquad-Kx)z  \\
          &\leq L+T + \lfloor\frac{T-KM-1}{Kx}\rfloor -z  \\
          &\leq L+T + \lfloor\frac{T-KM-1}{Kx}\rfloor -\zbr  \\
          &<0.
    \end{align*}

    We now analyze $(\eta -M+1)x^{-1}\bmod q \bmod z$ and we have \begin{align*}
      &(\eta -M+1)x^{-1}\bmod q \bmod z  \\
      &= t_1 + \lfloor\frac{t^\prime_2}{Kx}\rfloor + (t^\prime_2 \bmod Kx)z + L \bmod z \\
                                       &= L+ t_1 + \lfloor\frac{t^\prime_2}{Kx}\rfloor \bmod z.
    \end{align*}
    The modulo operation vanishes here as well as\begin{align*}
      L+ t_1 + \lfloor\frac{t^\prime_2}{Kx}\rfloor &\leq L+ T-1 + \lfloor\frac{T-KM-1}{Kx}\rfloor \\
                                                   &<\zbr.
          \end{align*}

      Thus, \begin{align*}
        (\eta -M+1)x^{-1}\bmod q \bmod z &= L+ t_1 + \lfloor\frac{t^\prime_2}{Kx}\rfloor \\
                                         &\geq L.
      \end{align*}

    Applying the same operation to $u_{k, \ell}$ we have that \begin{align*}
        (u_{k, \ell}-M+1) x^{-1} \bmod q \bmod z &= \kb z + \lb  \bmod z \\
                                         &= \lb.
      \end{align*}
      I.e., after applying the map $\psi \to (\psi -M+1)x^{-1} \bmod q \bmod z$ the $u_{k, \ell}$ are in $\setint{0}{L-1}$, whereas the transformed $\eta$ are at least $L$.

  \emph{Condition III) d)} requires that no entry of an antidiagonal $u_{k, \ell}$ appears in the non-antidiagonal entries of the top-left quadrant.
  Assume not, then there exists a $\mu \neq M-1$ such that\begin{align*}
   \kbp y + \lbp x + \mu &= \kb y + \lb x + M-1.
 \end{align*}
 We reduce both sides modulo $x=(M+1)$ to obtain\begin{align*} 
   \mu \bmod (M+1) &= M-1,
 \end{align*}
 contradicting $\mu \in \setint{0}{2M-2}$ or $\mu \neq M-1$.

 \emph{Condition IV)} is satisfied by \cite[Lemma~1]{hofmeister2025cat} since $\as$ and $\bs$ are arithmetic progressions and their common differences $x$ and $y=xz$ are coprime to $q=Ky-1$.
\end{proof}

\fi

\end{document}